\documentclass[final,leqno,onefignum,onetabnum,letterpaper]{siamltexmm}

\usepackage[USenglish]{babel} 
\usepackage[T1]{fontenc}

\usepackage{lmodern} 
\usepackage{color}

\usepackage{bmpsize}
\usepackage[caption=false]{subfig}

\usepackage{booktabs} 
\newcommand{\ra}[1]{\renewcommand{\arraystretch}{#1}}

\usepackage[textsize=tiny]{todonotes}

\usepackage[letterpaper,centering]{geometry}
\usepackage{amsmath}
\usepackage{amsfonts}
\usepackage{mhequ}

\usepackage{algorithm} 
\usepackage{algpseudocode}

\usepackage{mathtools}

\newcommand{\etal}{\textit{et al.\ }}

\DeclareMathOperator*{\argmax}{arg\,max}

\def \be{\begin{equs}}
\def \ee{\end{equs}}
\def \P{\mathbb{P}}
\def \E{\mathbb{E}}

\def \tmix{\tau_{\mathrm{mix}}}

\newtheorem{remark}[theorem]{Remark}
\newtheorem{defn}[theorem]{Definition}
\newtheorem{thm}{Theorem}

\graphicspath{{figures/}}

\author{Patrick R.\ Conrad\footnotemark[1] \and Andrew D.\ Davis\footnotemark[1] \and Youssef M.\ Marzouk\footnotemark[1] \and Natesh S.\ Pillai\footnotemark[2]  \and Aaron Smith\footnotemark[3]}

\definecolor{purple}{RGB}{129,15,124}
\newcommand{\change}[1]{#1} 

\title{Parallel local approximation MCMC for expensive models}

\begin{document}

\maketitle
\newcommand{\slugmaster}{\slugger{juq}{2018}{}}

\footnotetext[1]{Massachusetts Institute of Technology, Cambridge, MA 02139, USA, \texttt{\{prconrad,\ davisad,\ ymarz\}@mit.edu}. }
\footnotetext[2]{Harvard University, Cambridge, MA
  02138, USA, \texttt{pillai@stat.harvard.edu}.}
\footnotetext[3]{University of Ottawa, Ottawa, ON
K1N 7N5, Canada, \texttt{asmi28@uottawa.ca}.}

\renewcommand{\thefootnote}{\arabic{footnote}}

\begin{abstract}
Performing Bayesian inference via Markov chain Monte Carlo (MCMC) can be exceedingly expensive when posterior evaluations invoke the evaluation of a computationally expensive model, such as a system of partial differential equations. In recent work \cite{Conrad2014}, we described a framework for constructing and refining \textit{local approximations} of such models during an MCMC simulation. These posterior--adapted approximations harness regularity of the model to reduce the computational cost of inference while preserving asymptotic exactness of the Markov chain. Here we describe two extensions of that work. 
\change{First,} we prove that samplers running in parallel can collaboratively construct a shared posterior approximation while ensuring ergodicity of each associated chain, providing a novel 
opportunity for exploiting parallel computation in MCMC.
\change{Second,} focusing on the Metropolis--adjusted Langevin algorithm, we describe how a proposal distribution can successfully employ gradients and other relevant information extracted from the approximation.
We investigate the practical performance of our approach using two challenging inference problems, the first in subsurface hydrology and the second in glaciology. Using local approximations constructed via parallel chains, we successfully reduce the run time needed to characterize the posterior distributions in these problems from days to hours and from months to days, respectively, dramatically improving the tractability of Bayesian inference.  \end{abstract}



\pagestyle{myheadings}
\thispagestyle{plain}

\maketitle

\pagestyle{plain} 

\section{Introduction}

Markov chain Monte Carlo (MCMC) is a powerful tool for performing Bayesian inference, but can be computationally prohibitive in many settings, especially when posterior density evaluations involve a computationally expensive step. For instance, applications in the physical sciences often require partial differential equation \textit{forward models}, evaluated using numerical solvers with nontrivial run times. When these solvers must be invoked with each posterior evaluation, direct sampling with MCMC can become intractable.

\change{To reduce this computational burden, 
a standard approach is} to construct an approximation or ``surrogate'' of the forward model or likelihood function, and then to sample from (or otherwise characterize) the posterior distribution induced by this approximation \cite{Sacks1989, Kennedy2001, Rasmussen2003, Santner2003, Marzouk2007, Marzouk2009a, Bliznyuk2012, Joseph2012, Li2012,Chen2015,CuiWM2016}.  Although such approaches can be quite effective at reducing computational cost, they may be difficult to use in practice---in part because they separate the construction of the surrogate from the subsequent inference procedure. 
%
Approximation of the forward model biases posterior expectations \cite{cotter2010} in a way that cannot easily be quantified. 
It is then difficult to decide how much computational effort should be devoted to surrogate construction, and how to balance the resulting biases with the statistical errors of posterior sampling. 
%
%
Alternatives such as delayed-acceptance MCMC \cite{Christen2005,cui2011wrr} yield asymptotically exact sampling, but surrender potential speedups by requiring at least one evaluation of the forward model for each accepted sample.
In recent work \cite{Conrad2014}, we demonstrated that surrogate construction and posterior exploration can instead be \textit{joined}, yielding a framework for incrementally and infinitely refining a surrogate during MCMC sampling. This framework allows the approximation to be tailored to the problem---e.g., made most accurate in regions of high posterior probability---while guaranteeing that the associated Markov chain asymptotically samples from the \textit{exact} posterior distribution of interest. Empirical studies on problems of moderate dimension showed that the number of expensive posterior evaluations per MCMC step can be reduced by orders of magnitude, with no discernable loss of accuracy in posterior expectations. 

This work describes two key extensions of the framework in \cite{Conrad2014}. First, we show that our approximation scheme enables a novel type of MCMC parallelism: concurrent chains can collaboratively develop a \textit{shared} approximation. Effectively exploiting parallel computation in MCMC is often challenging because the core algorithm is inherently sequential, but our strategy directly deploys parallel resources to address the key performance bottleneck: the cost of repeatedly running the forward model.

Second, while our previous work showed how to build a convergent approximation of the target probability density, it did not support the idea of using this approximation to construct a proposal distribution. MCMC performance is highly dependent on the choice of proposal, but sophisticated proposals, such as the Metropolis-adjusted Langevin algorithm (MALA) and its manifold variants \cite{Girolami2011}, can be expensive to apply because they require gradients (and possibly higher derivatives) of the forward model. This derivative information is often expensive or impossible to compute directly, but is trivial to extract from an approximation. Intuitively, it should then be possible to use our approximation framework to greatly reduce the costs of such proposals. Here we do exactly that, extending our previous theoretical results to show that the Monte Carlo estimates obtained by our algorithm converge to the correct value, as long as the convergence of our approximation to the target distribution yields convergence of the associated approximate Markov transition kernel in a suitably strong norm. As an example, we show how to use simplified manifold MALA within our local approximation scheme, and prove that the resulting stochastic process is convergent in a representative case.

Finally, we construct two 
inference problems that are representative of interesting scientific queries, that involve computationally expensive forward models (such that na\"{i}ve use of the model in sampling would take days or months), and that have nontrivial posterior structure which must be characterized using MCMC. 
The first is a problem in groundwater hydrology, where a subsurface conductivity field is inferred from observations of tracer transport; it is a more complex and realistic version of the linear elliptic PDE inverse problem \cite{Dashti2011}, combining an elliptic equation for the hydraulic head with another PDE governing tracer dispersion \cite{Dupuit1863, Matott2012}. 
The second problem is drawn from glaciology: here we infer the basal friction parameters of a shallow-shelf ice stream model \cite{MacAyeal1989, MacAyeal1993, MacAyeal1997} from observations of surface ice velocity. 
%
Our numerical experiments evaluate MCMC efficiency, accuracy, and wallclock time, and benchmark the parallel performance of our algorithms. Results demonstrate strong performance of our approach; for example, inference in the ice stream model becomes tractable, with the time needed to characterize the posterior reduced from roughly two months to just over a day. 

The remainder of this paper is organized as follows. Section~\ref{sec:review} reviews the basic algorithmic framework of local approximation (LA) MCMC. Section~\ref{sec:parallel} presents and analyzes the shared construction of approximations for parallel MCMC. Section~\ref{sec:proposal} describes the use of local approximations in the proposal scheme, and Section~\ref{sec:numerical} describes our numerical experiments. Proofs of the main convergence results, along with certain algorithmic details, are deferred to the appendices.


\section{Review of local approximation MCMC}
\label{sec:review}

We are interested in Bayesian inference problems with posterior densities of the form
\begin{equation*}
p(\theta|\mathbf{d} ) \propto \ell(\theta | \mathbf{d}, \mathbf{f})p(\theta),
\end{equation*}
for parameters $\theta \in \Theta \subseteq \mathbb{R}^d$, data $\mathbf{d} \in \mathbb{R}^n$, a forward model $\mathbf{f}: \Theta \rightarrow \mathbb{R}^n$, and probability densities specifying the prior $p(\theta)$ and likelihood function $\ell$. The forward model may enter the likelihood function in various ways. For instance, if $\mathbf{d} = \mathbf{f}(\theta) + \eta$, where $\eta$ represents some measurement or model error with probability density $p_{\eta}$, then $\ell(\theta |  \mathbf{d}, \mathbf{f}) = p_{\eta} ( \mathbf{d} - \mathbf{f}(\theta) )$.

Assume that the forward model is both \textit{computationally expensive} and a \textit{black box}, so that we cannot inspect or modify it. In this setting, standard approaches to MCMC are likely to be limited by the computational expense of evaluating the forward model at every step of the chain. Our approach addresses this cost by storing the results of each model evaluation in a set $\mathcal{S}_t \coloneqq \{(\theta_i, \mathbf{f}(\theta_i))\}_{i=1}^{n_t}$ and reusing them. The stochastic process $\{ \theta_{t} \}_{t \geq 0}$ proposed in \cite{Conrad2014} evolves by drawing new points from some proposal kernel $q$ and accepting or rejecting the proposed move according to an approximation of the forward model, $\tilde{\mathbf{f}}_t$, constructed from the set $\mathcal{S}_{t}$. During the simulation of this process, the algorithm carefully chooses new points at which to run the forward model, enlarging $\mathcal{S}_t$ and thus improving $\tilde{\mathbf{f}}_t$; \change{we refer to enlargement of $\mathcal{S}_t$ as ``refinement.''}
Intuitively, it would seem that if $\tilde{\mathbf{f}}_t$ converges to $\mathbf{f}$ in an appropriate sense, then the sequence $\{ \theta_{t} \}_{t \geq 0}$ might asymptotically behave like the usual Metropolis-Hastings chain with proposal $q$ and forward model $\mathbf{f}$. Indeed, the algorithm we constructed in \cite{Conrad2014} has these properties. 

We obtain a converging sequence of approximations $\tilde{\mathbf{f}}_t$ by constructing the approximation locally---that is, constructing $\tilde{\mathbf{f}}_t(\theta)$ using \change{only the elements of $\mathcal{S}_t$ whose input values $\theta_i$ lie within a distance $R$ of $\theta$. The radius $R$ is selected so that this subset contains a fixed number of points $N$. The value of $N$ depends on the functional form of the approximation; for instance, if  $\tilde{\mathbf{f}}_t$ is a local quadratic approximation, we need at least $(d+1)(d+2)/2 \eqqcolon N_{\text{def}}$ points to fully determine its coefficients.\footnote{\change{In practice, we often select $N = \sqrt{d} N_{\text{def}}$ to improve the conditioning of the associated least squares system. More details are given in \cite{Conrad2014}.}}}
Local approximations are relatively straightforward to analyze in that they typically converge whenever the sample set $\mathcal{S}_t$  becomes denser, \change{thus} allowing $R \to 0$. 
(Regularity conditions on $\mathbf{f}$ sufficient for convergence in the case of local polynomial approximations, for example, are given in \cite{Conn2009}.) These general conditions for convergence allow us to promote efficiency by aggressively tailoring $\mathcal{S}_t$ during sampling, while still maintaining asymptotic exactness of the overall MCMC. The resulting algorithm is straightforward to use, since its adaptivity allows users to treat it much like standard adaptive MCMC algorithms: the behavior of the chain can be monitored for convergence, which in our case reflects both the exploration of the posterior and the convergence of the approximation. Our work thus differs from previous efforts using global approximations to accelerate inference \cite{Marzouk2007,Bliznyuk2012,Joseph2012,Santner2003}, where the entire set $\mathcal{S}_t = \mathcal{S}_0$ is constructed as a preprocessing step and is used to build a single high-order approximation. In these methods it is difficult to choose how many samples $\mathcal{S}_0$ should contain or how to monitor the accuracy of the overall sampling.

An illustration of the algorithm is given in Figure \ref{fig:refinementCartoon}. At early times, the samples are sparse, leading to local models constructed over large regions, depicted by large balls, rendering them relatively inaccurate. As MCMC progresses, refinements increase the density of the sample set in regions of high posterior probability, shrinking the local neighborhoods and increasing the quality of approximations. Model runs do not lie on any structured grid and are generally contained within regions of the parameter space that are relevant to the inference problem, thus enhancing efficiency whenever the posterior is concentrated.

\begin{figure}[htb]
\centering
\subfloat[Early times.]{
\def\svgwidth{0.4\textwidth}
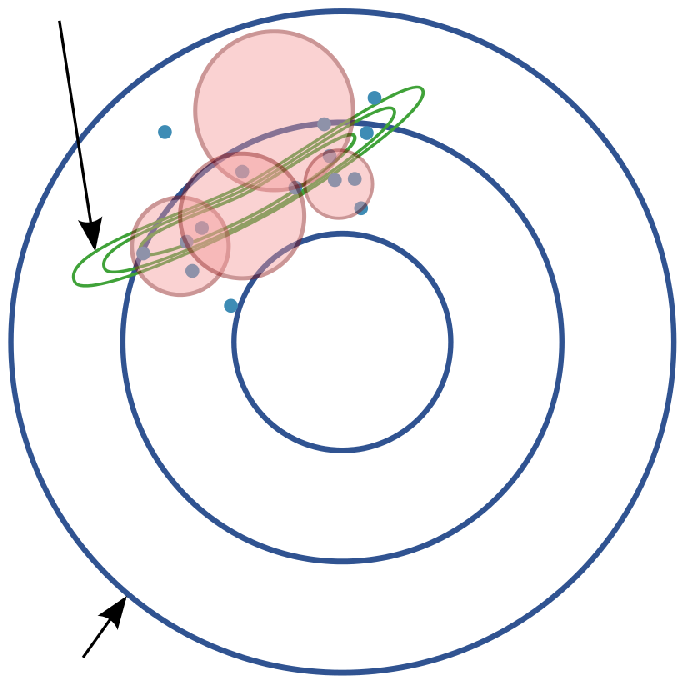}
\subfloat[Late times.]{
\def\svgwidth{0.4\textwidth}
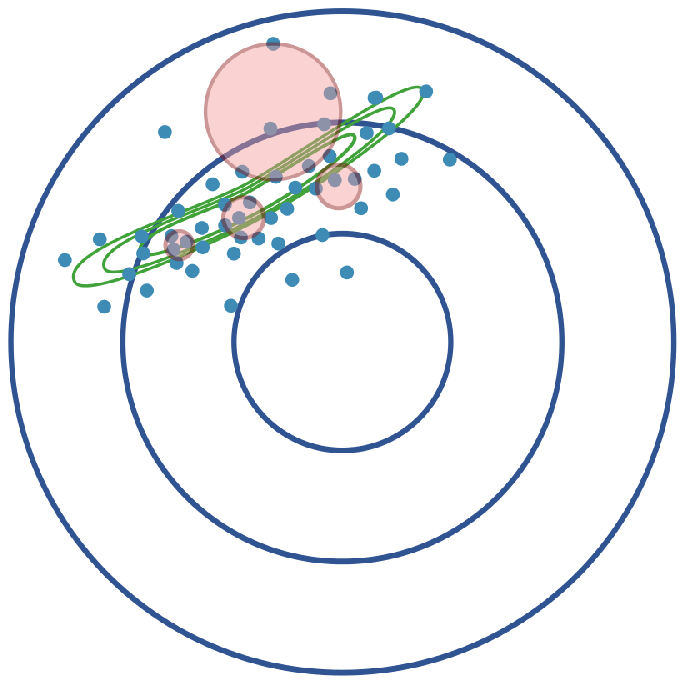}
\caption{Schematic of the behavior of local approximation MCMC. The balls are centered at locations where local approximations might be \change{evaluated, and their radii are chosen to contain the $N$ nearest points, used to build the approximation. The accuracy of a local approximation generally increases as this ball size shrinks. At early times, the sample set $\mathcal{S}_t$ is sparse and thus local approximations are built over relatively large balls, such that their accuracy is limited. At later times, refinements enrich the sample set in regions of high posterior probability, allowing the balls to shrink and the approximations to become more accurate.}}
\label{fig:refinementCartoon}
\end{figure}

We now review a sketch of our approximate MCMC algorithm, given in Algorithm \ref{alg:algSketch}. \change{Please see  Appendix~\ref{apx:code} and Algorithm~\ref{alg:algOverviewGeneral} for a more complete description of the algorithm; additional details can be found in our previous work \cite{Conrad2014}.} 
The stochastic process $\{ \theta_{t} \}_{t \geq 0}$ is produced by the method \textsc{RunChain}, which applies the transition kernel $K_t$ repeatedly. The transition kernel is provided with the current state of the chain $\theta_t$; the current set of samples $\mathcal{S}_t$; the inference problem, as defined by $\ell$, $\mathbf{d}$, $p$, and $\mathbf{f}$; and a symmetric translation--invariant \change{proposal distribution $q$.} 
The kernel uses the current point of the chain, $\theta^-$, to draw a proposal, $\theta^+$. It forms local approximations near these points, $\tilde{\mathbf{f}}^+$ and $\tilde{\mathbf{f}}^-$, respectively, based on nearby samples contained in $\mathcal{S}_t$. Next, it computes the acceptance probability $\alpha$ in the usual way, substituting the approximations for the true forward model. Then, the algorithm optionally refines the sample set by choosing a new point $\theta^\ast$ and running the forward model at that location.

\algrenewcomment[1]{\hfill\makebox[0.33\linewidth][l]{\(\triangleright\) #1}}

\begin{algorithm}
\caption{Sketch of approximate Metropolis-Hastings algorithm}
\label{alg:algSketch}
\begin{algorithmic}[1]
\Procedure{\textsc{RunChain}}{$\theta_{1}, \mathcal{S}_{1}, \ell, \mathbf{d}, p, \mathbf{f}, q, T$}
\For{$t = 1 \ldots T$}
	\State  $(\theta_{t+1}, \mathcal{S}_{t+1}) \gets K_t(\theta_t, \mathcal{S}_t, \ell, \mathbf{d}, p, \mathbf{f}, q)$
\EndFor
\EndProcedure
\Statex
\Procedure{$K_t$}{$\theta^-, \mathcal{S},  \ell, \mathbf{d}, p, \mathbf{f}, q$}
	\State Draw proposal $\theta^+ \sim q(\theta^-, \cdot)$ 

\State Compute approximate models $\tilde{\mathbf{f}}^+$ and $\tilde{\mathbf{f}}^-$, \label{alg:algSketch:construct} valid near $\theta^+$ and $\theta^-$, respectively.
\State Compute acceptance probability \label{alg:algSketch:alpha}
 $\alpha \gets \min \left(1,\frac{\ell(  \theta | \mathbf{d} ,\tilde{\mathbf{f}}^+)p(\theta^+)}{\ell(\theta  | \mathbf{d}  ,\tilde
{\mathbf{f}}^-)p(\theta^-)} \right)$ 
\If{approximation needs refinement near $\theta^-$ or $\theta^+$}
\State Select new point $\theta^\ast$ and grow $\mathcal{S} \gets \mathcal{S} \cup (\theta^\ast, \mathbf{f}(\theta^\ast))$. Repeat from Line \ref{alg:algSketch:construct}.
\Else
\State Draw $u \sim \text{Uniform}(0,1)$. If $u < \alpha$, \textbf{return} $(\theta^+, \mathcal{S})$; else \textbf{return} $(\theta^-, \mathcal{S})$.
\EndIf
\EndProcedure

\end{algorithmic}
\end{algorithm}

Choosing when and where to refine $\mathcal{S}_t$ is critical to the performance of the overall algorithm. We combine two criteria to decide when to refine the approximation. First, the approximation is refined near $\theta^-$ or $\theta^+$ with equal probabilities $\beta_t$, such that the expected number of refinements diverges as $t \to \infty$. This criterion is sufficient for convergence of the algorithm, as detailed in \cite{Conrad2014}. The sequence $(\beta_t)$ may be difficult to tune in practice, however. Thus we complement the random refinement criterion with a cross--validation strategy that triggers refinement whenever the estimated error in the acceptance probability $\alpha$ \change{(due to the approximation of $\mathbf{f}$)} appears too large. This latter threshold for refinement is tightened with increasing $t$, pushing the approximation to improve as the chain lengthens. Although the cross--validation criterion is not sufficient for convergence of the algorithm, it appears efficient in practice, and we use it in conjunction with the random refinement strategy.
When refinement is needed near either of $\theta^-$ or $\theta^+$, we do not simply run the model at that point, since doing so would introduce clusters into $\mathcal{S}_t$, degrading the quality of local approximations. Instead, we use a local space-filling design strategy to choose a distinct but nearby point $\theta^\ast$ at which to run the model.

\section{Sharing local approximations for parallel MCMC}
\label{sec:parallel}


The na\"{i}ve approach to parallelizing MCMC is simply to run several independent chains in parallel. Although running parallel chains facilitates useful convergence diagnostics \cite{Cowles1996, Brooks1998}, practical scaling in highly parallel environments is limited because of the serial nature of MCMC and the replication of transient behavior across multiple chains \cite{Rosenthal2000}.

More sophisticated strategies for parallel MCMC exchange information between the chains, for example by proposing moves to states discovered by other chains \cite{Craiu2009}. Population MCMC algorithms explore a family of tempered distributions with parallel chains, so that swapping states between the chains can provide long-range moves \cite{Cappe2004}. These techniques attempt to improve the mixing time of the Markov chain, and when successful, may provide superior performance to the na\"{i}ve parallelization \cite{green2015bayesian}. Other constructions, e.g.,  \cite{Calderhead2014}, propose multiple points in parallel and try to make use of all these points in determining subsequent steps of a single chain.

Any of these parallel approaches requires repeated evaluations of the forward model, however, which can dominate the overall cost of the algorithm. If multiple copies of Algorithm \ref{alg:algSketch} are run in parallel, a natural idea is to allow them to collaborate by sharing a common set of evaluations $\mathcal{S}_t$. That is, whenever one chains performs refinement, the result is shared asynchronously with all the chains; hence each chain receives additional model evaluations ``for free.'' 
Since the limiting computational cost in our context lies in constructing $\mathcal{S}_t$, parallelizing this process should directly impact the real-world performance of the sampler \change{during the stationary and even the transient phases of the chains.}
\change{With regard to the latter point, we note that parallelizing $\mathcal{S}_t$ can reduce the number of model evaluations that are triggered by each individual chain during its initial transient phase.} 

Although it should be straightforward to combine \change{the parallel construction of $\mathcal{S}_t$ with the other parallelization strategies described above}, we leave that as future work.

\subsection{Convergence of the parallel algorithm}

\change{Recall that our local approximation MCMC algorithm is detailed in Appendix~\ref{apx:code}; see, in particular, Algorithm \ref{alg:algOverviewGeneral}. Below we will show that the sufficient conditions for convergence of a single-chain version of Algorithm \ref{alg:algOverviewGeneral}, as described in \cite{Conrad2014} and reproduced below in Definition~\ref{DefSomeAssumptions}, are also sufficient conditions for the convergence of the parallel version.} The arguments given in \cite{Conrad2014} are straightforward to extend because \change{we have chosen} conditions where enlarging the sample set $\mathcal{S}_t$ is always helpful; thus the additional refinements contributed by parallel chains cannot hinder convergence. Rather than repeating the entire discussion of convergence from that paper, here we merely extend the simplest and weakest convergence result---for a single chain on a compact state space \cite[Theorem 3.4]{Conrad2014}---to the case of parallel chains. We refer the reader to \cite[Theorem 3.3]{Conrad2014} for related conditions and a treatment of non-compact state spaces that can similarly be extended to the parallel case.

We require some notation before stating the result. \change{Let $\mathcal{L}(X)$ denote the distribution} of a random variable $X$. For fixed $\epsilon > 0$, we say that $\mathcal{S} \subset \Theta$ is an $\epsilon$-cover of $\Theta$ if $\sup_{\theta \in \Theta} \min_{s \in \mathcal{S}} \| \theta - s \|_{2} < \epsilon$. We note that if Steps 13--21 and Step 23 are removed from Algorithm \ref{alg:algOverviewGeneral}, and all references to $\mathcal{S}_t$ are replaced by a reference to a single set $\mathcal{S}$, then the sequence $\{ \theta_{t} \}_{t \geq 0}$ constructed by running the modified algorithm is a Markov chain. We use the $\mathcal{S}$ subscript to denote all approximate objects associated with this Markov chain (e.g., $K_{\mathcal{S}}$ is the associated transition kernel, $r_{\mathcal{S}}$ is the proposal function from Step 9 of Algorithm \ref{alg:algOverviewGeneral} and $q_{\mathcal{S}}$ is the associated proposal density, $p_{\mathcal{S}}:= \ell(\theta \vert \mathbf{d}, \tilde{\mathbf{f}}) p(\theta)$ is the approximation to $\ell(\theta | \mathbf{d}, \mathbf{f}) p(\theta)$ used in Step 11 of Algorithm \ref{alg:algOverviewGeneral}, and $\alpha_{\mathcal{S}}$ is the associated acceptance probability).  Similarly, $K_{\infty}$, $r_{\infty}$,  $q_{\infty}$, $p_{\infty}$, and $\alpha_{\infty}$ are the values of these objects for the Markov chain with the same proposal kernel as in Algorithm \ref{alg:algOverviewGeneral} and with the \emph{correct} posterior distribution as its target distribution. Finally, define $\pi(\theta) := p(\theta) \ell(\theta | \textbf{d}, \textbf{f})/Z$, where $Z$ is a normalization constant. 
Our simple result makes the following assumptions:

\smallskip

\begin{defn} [Sufficient conditions for convergence] \label{DefSomeAssumptions}

\begin{enumerate}
\item The state space $\Theta$ is compact.
\item The proposal $q(\theta, \cdot \,  \vert \, \mathbf{f} \, ) = q(\theta, \cdot)$ does not depend on $\mathbf{f}$, 
and both the proposal distribution $q(\theta, \cdot)$ and target distribution $p(\cdot \vert \, \mathbf{d} \, )$ have $C^{\infty}$ densities that are bounded away from zero uniformly in $\theta$.
\item The sequence of parameters $\{ \beta_{t}\}_{t \in \mathbb{N}}$ used in Algorithm \ref{alg:algOverviewGeneral} are of the form $\beta_{t} \equiv \beta > 0$; any sequence $\{ \gamma_{t} \}_{t \geq 0}$ is allowed.
\item The approximation of $\log p(\theta | \mathbf{d})$ is made via quadratic interpolation on the $N = (d+1)(d+2)/2$ nearest points.
\item The sub-algorithm \textsc{RefineNear} is replaced with:
\be 
\textsc{RefineNear}(\theta, \mathcal{S}) = \textbf{return}( \mathcal{S} \cup \{ (\theta, f(\theta)) \} ).
\ee 
\item We fix a constant $0 < \lambda < \infty$. In Step 15, immediately before the word \textbf{then}, we add `\textbf{or}, for $\mathcal{B}(\theta^{+}, R)$ as defined in the subalgorithm $\textsc{LocApprox}(\theta^{+}, \mathcal{S}, \emptyset)$ used in Step 8, the collection of points $\mathcal{B}(\theta^{+}, R) \cap \mathcal{S}$ is not $\lambda$-poised.' We add the same check, with $\theta^{-}$ replacing $\theta^{+}$ and `Step 10' replacing `Step 8,' in Step 17. The concept of poisedness is defined in \cite{Conn2009}.
\end{enumerate}
\end{defn}

\medskip

The following result extends Theorem 3.4 of \cite{Conrad2014} to parallel chains.

\smallskip

\begin{thm} [Convergence with parallel chains]
\ Let $\{X_{t}^{(i)}\}_{t \geq 0, \, 1 \leq i \leq n}$ be the $n$ stochastic processes obtained in a parallel run of Algorithm \ref{alg:algOverviewGeneral} with $n$ chains, and assume that the algorithm parameters satisfy Assumption \ref{DefSomeAssumptions}. Then, for all $1 \leq i \leq n$, 
\be 
\lim_{t \rightarrow \infty} \| \mathcal{L}(X_{t}^{(i)}) -\pi \|_{TV} = 0.
\ee 
\end{thm}
\begin{proof}\renewcommand{\endproof}{}
The proof of Lemma B.3 of \cite{Conrad2014} holds exactly as stated, with the proof as given. The remainder of the proof of Theorem 3.4  from \cite{Conrad2014} holds for $\{ X_{t}^{(i)} \}_{t \geq 0}$, for each fixed $1 \leq i \leq n$, with the following modifications:
\begin{itemize}
\item[(i)] The chain $\{ X_{t} \}_{t \geq 0}$ should be replaced by $\{ X_{t}^{(i)} \}_{t \geq 0}$ wherever it appears; and
\item[(ii)] The auxillary process associated with $\{ X_{t}^{(i)} \}_{t \geq 0}$ is $\left \{ \left ( \mathcal{S}_{t}, X_{t}^{(1)},\ldots,X_{t}^{(i-1)}, X_{t}^{(i+1)}, \ldots, X_{t}^{(n)}  \right ) \right \}_{t \geq 0}$, rather than $\{ \mathcal{S}_{t} \}_{t \geq 0}$. {\color{header1}\rule{1.5ex}{1.5ex}}
\end{itemize} \end{proof}
\medskip
We emphasize that this proof of the convergence of $\{ X_{t}^{(i)} \}_{t \geq 0}$ is completely indifferent to the points that are added to $\mathcal{S}_{t}$ by the other chains $\{ X_{t}^{(j)} \}_{t \geq 0}$, $j\neq i$. 
\medskip

\begin{remark}[Do parallel chains always work?]
\ Although our sufficient conditions for convergence carry over to the parallel case, it is natural to ask whether there are any problems that are not covered by our current theory---i.e., where, having departed from the sufficient conditions of Definition~\ref{DefSomeAssumptions}, the single-chain algorithm still converges but the parallel algorithm does not. We conjecture that the answer is no, but are unable to prove it. 

To explain the difficulty in proving this conjecture, note that all the proofs of sufficient conditions for convergence given in \cite{Conrad2014} apply as stated to the parallel version of Algorithm \ref{alg:algOverviewGeneral} because they proceed by proving the following critical steps:
\begin{enumerate}
\item Due to minorization conditions (e.g., the second condition of Assumption \ref{DefSomeAssumptions}), for any $\epsilon > 0$ the set $\mathcal{S}_{t}$ will be an $\epsilon$-cover of $\Theta$ for all $t$ sufficiently large.
\item The distance $\Vert \, K_\infty(\theta_t, \cdot) - K_{\mathcal{S}_t}(\theta_t, \cdot ) \, \Vert_{\text{TV}}$ between a single `step' of Algorithm \ref{alg:algOverviewGeneral} and the step that would be made by the true transition kernel $K_\infty$ can be made arbitrarily small by making $\mathcal{S}_t$ an $\epsilon$-cover of $\Theta$ for $\epsilon$ sufficiently small.

\end{enumerate}

In particular, under \change{Assumption \ref{DefSomeAssumptions}}, adding points to $\mathcal{S}_{t}$ cannot hurt the convergence of $\{ X_{t}^{(i)} \}_{t \geq 0}$ very much, because adding points to an $\epsilon$-cover always results in a set that is still an $\epsilon$-cover. For a sufficiently \change{broader} class of Metropolis-Hastings chains, however, it is {not} true that $K_{\mathcal{S}}$ is close to $K_{\infty}$ whenever $\mathcal{S}$ is an $\epsilon$-cover, and in particular it is possible to {add} points to $\mathcal{S}$ while simultaneously making an approximation worse. This possibility of mal-adaption is what makes adaptive algorithms difficult to study, and prevents us from making the stronger claim that the parallel algorithm is convergent under \emph{every} possible condition where the single-chain algorithm is.
\end{remark}



\section{Local approximations and approximating the proposal}
\label{sec:proposal}

We now show how the transition kernel of our approximate MCMC scheme can use the current approximation not only to evaluate the acceptance probability, but also to construct a proposal distribution. This development enables a much wider range of Metropolis-Hastings proposals to be used with expensive models, and in particular allows gradient- and Hessian-driven proposals to be used in a setting where derivatives of $\mathbf{f}$ cannot be directly evaluated. 
%
We proceed by recalling the Metropolis-adjusted Langevin algorithm (MALA) algorithm and explaining how to adapt local approximations to this proposal scheme. Next, we prove a general result that our modified algorithm is still convergent as long as the good properties of the approximation are transferred into good approximation of the overall kernel.  We conclude by showing that the result applies in the representative case of manifold MALA.

 \subsection{Simplified manifold Metropolis-adjusted Langevin algorithm (mMALA)}
 
The simplified manifold Metropolis-adjusted Langevin algorithm (mMALA) \cite{Girolami2011} is a recent method for constructing proposals adapted to the local geometry of the target distribution. This method is also closely related to the preconditioning performed in the stochastic Newton method \cite{Martin2012}.  The mMALA proposal is derived by explicitly discretizing a Langevin diffusion with stationary distribution $p(\theta | \mathbf{d})$, leading to
%
\be
q(\theta,\theta' | \mathbf{f}) = \mathcal{N}\left(\theta'; \theta + \frac{\epsilon}{2} M(\theta) \nabla_\theta \log{ \left (\ell(\theta | \mathbf{d}, \mathbf{f}) p(\theta) \right )}, \epsilon  M(\theta)\right),
\label{eqn:malaProposal}
\ee 
for integration step size $\epsilon$ and position-dependent symmetric positive definite (SPD) mass matrix $M(\theta)$, which we may treat as a preconditioner. \change{``Preconditioning'' in this context amounts to rescaling the parameter space, e.g., to make the distribution (locally) more isotropic.}
We use the notation $q(\theta,\theta' \vert \mathbf{f})$ to emphasize the dependence of the proposal on the forward model. The corresponding acceptance ratio is
\be
\alpha(\theta,\theta') = \min \left(1, \frac{\ell(\theta' | \mathbf{d}, \mathbf{f}) p(\theta') q(\theta',\theta \vert \mathbf{f})}{\ell(\theta | \mathbf{d}, \mathbf{f}) p(\theta) q(\theta,\theta' \vert \mathbf{f})} \right).
\ee

We are relatively unconstrained in our choice of preconditioner, as long as it is SPD. Standard MALA corresponds to choosing the identity matrix, $M(\theta) = I$. Simplified manifold MALA (mMALA) \cite{Girolami2011}, on the other hand, chooses the mass matrix to reflect a Riemannian metric induced by the posterior distribution:
\begin{eqnarray*}
M(\theta) &=&  \left[- \mathbb{E}_{\mathbf{d}|\theta} \left( \nabla^2_\theta \log \ell(\theta | \mathbf{d}, \mathbf{f}) \right)  -  \nabla^2_\theta \log p(\theta) \right]^{-1}.
\end{eqnarray*}
The inverse of this matrix is the expected Fisher information plus the negative Hessian of the log-prior density. In general, computing the expected Fisher information is not trivial, but it is relatively simple for Gaussian likelihoods, e.g., 
\begin{eqnarray*}
\ell(\theta | \mathbf{d}, \mathbf{f})  &=& \mathcal{N}(\mathbf{d}; \mathbf{f}(\theta), \Sigma_\ell), 
\end{eqnarray*}
with some \change{prescribed} covariance matrix $\Sigma_\ell \in \mathbb{R}^{n \times n}$. If we also have a Gaussian prior, $p(\theta) = \mathcal{N}(\theta; \mu, \Sigma_p)$, with covariance $\Sigma_p \in \mathbb{R}^{d \times d}$ and mean vector $\mu \in \mathbb{R}^d$, then
\be
M^{-1}(\theta) = J(\theta)^\top \Sigma^{-1}_\ell J(\theta) + \Sigma_p^{-1},
\ee
where $J(\theta) := \nabla_\theta \mathbf{f}(\theta) \in \mathbb{R}^{n \times d}$. Girolami \etal \cite{Girolami2011} observe that choosing the preconditioner in this manner can dramatically improve the performance of MALA. Yet even standard MALA can be difficult to apply in practice because the necessary derivatives must be computable and inexpensive; the manifold variant uses Jacobians of the forward model, which are typically even more challenging to obtain. Adapting mMALA and similar proposals to use local approximation is therefore particularly interesting, as approximations can cheaply provide these derivatives. 

\subsection{Modifying the algorithm}
The key challenge in extending Algorithm \ref{alg:algSketch} to mMALA (and similar proposals) is to allow simultaneous use of the approximation within the proposal and the acceptance probability. Algorithm \ref{alg:algSketchMala} shows the three required changes.  
Two modifications are trivial: we restore the proposal distribution to its usual place in the acceptance probability, to account for the non-symmetric proposal; and we provide the proposal with the approximate forward model $\tilde{\mathbf{f}}$. 

The third step is more subtle, introducing a coupling construction to allow model refinement to proceed safely. Note that in Algorithm \ref{alg:algSketch}, refinement only recomputes the acceptance probability; the proposed point is held fixed. Hence, exactly one proposal is made per step, even though an inaccurate approximation might cause the algorithm to seek further information before deciding whether that proposal can be accepted. Allowing a new proposal to be generated upon refinement would bias the chain away from regions with inaccurate approximations (equivalently, towards regions where the approximation appears accurate), which is clearly undesirable.

This difficulty can be resolved by coupling the approximate kernel $K_t$ to the kernel associated with the true model, $K_\infty$. We accomplish this coupling by fixing the realization of the random variable used to generate the proposal, but allowing the proposal to be recomputed if the model is refined. (See \cite{propp1996exact,fill2000randomness} for other algorithms that re-use randomness to avoid bias, and \cite{steinsaltz1999locally} for a typical use of this idea in a theoretical paper.) Specifically, construct a deterministic function $r(\theta, \mathbf{z}, \mathbf{f})$ such that drawing a random vector $\mathbf{z} \sim \mathcal{N}(0, I)$\footnote{We choose a vector of independent standard Gaussians for convenience and without loss of generality, but in practice other distributions for $\mathbf{z}$ may be more convenient.} and then computing $\theta'=r(\theta, \mathbf{z}, \mathbf{f})$ is equivalent to drawing $\theta' \sim q(\theta, \cdot | \mathbf{f})$. The modified algorithm holds $\mathbf{z}$ fixed under refinement, recomputing $\theta^+$ as needed. 
In the case of standard Metropolis-Hastings proposals, this coupling strategy reduces to our original approach.  
This coupling construction ensures that the magnitude of any perturbation to the proposed point $\theta^+$ induced by refinement vanishes as $t \to \infty$.

In the case of simplified manifold MALA, the proposal will be a Gaussian distribution, $q(\theta, \theta' \vert \mathbf{f}) = \mathcal{N}\left (\mu_q(\theta, \mathbf{f}), \Sigma_q(\theta, \mathbf{f}) \right )$, for some position- and model-dependent \change{mean $\mu_q$ and covariance $\Sigma_q$}, and hence $r = \mu_q(\theta, \mathbf{f}) + \Sigma_q^{1/2}(\theta, \mathbf{f}) \mathbf{z}$. The rest of the algorithm is updated naturally, including the inclusion of the proposal into the cross-validation criterion. The resulting approach is summarized in Algorithm~\ref{alg:algSketchMala}. For brevity, we defer precise pseudocode to Algorithm~\ref{alg:algOverviewGeneral} in Appendix \ref{apx:code}. 

\algrenewcomment[1]{\hfill\makebox[0.33\linewidth][l]{\(\triangleright\) #1}}

\begin{algorithm}
\caption{Sketch of approximate Metropolis-Hastings algorithm with general proposals}
\label{alg:algSketchMala}
\begin{algorithmic}[1]
\Procedure{$K_t$}{$\theta^-, \mathcal{S},  \ell, \mathbf{d}, p, \mathbf{f}, r, q$}
\State Draw $\mathbf{z}_t \sim \mathcal{N}(0,I)$ 
\State Construct  $\tilde{\mathbf{f}}^-$ \label{alg:algSketchMala:construct}
\State Compute $\theta^+ = r(\theta^-, \mathbf{z}_t, \tilde{\mathbf{f}}^-)$ 
\State Construct  $\tilde{\mathbf{f}}^+$ 

\State Compute acceptance probability \label{alg:algSketchMala:alpha}
 $\alpha \gets \min \left(1,\frac{\ell(  \theta^+ | \mathbf{d} ,\tilde{\mathbf{f}}^+)p(\theta^+)q(\theta^+, \theta^- | \tilde{\mathbf{f}}^+)}{\ell(\theta^-  | \mathbf{d}  ,\tilde
{\mathbf{f}}^-)p(\theta^-)q(\theta^-, \theta^+ | \tilde{\mathbf{f}}^-)} \right)$ 
\If{approximation needs refinement near $\theta^-$ or $\theta^+$}
\State Select new point $\theta^\ast$ and grow $\mathcal{S} \gets \mathcal{S} \cup (\theta^\ast, \mathbf{f}(\theta^\ast))$. Repeat from Line \ref{alg:algSketchMala:construct}. 
\Else
\State Draw $u \sim \text{Uniform}(0,1)$. If $u < \alpha$, \textbf{return} $(\theta^+, \mathcal{S})$, else \textbf{return} $(\theta^-, \mathcal{S})$.
\EndIf
\EndProcedure

\end{algorithmic}
\end{algorithm}

\subsection{Convergence analysis}

We now provide a convergence result for Algorithm \ref{alg:algOverviewGeneral}. Some technical definitions and the proofs from this section may be found in Appendix \ref{apx:proofs}. 

The general idea is to show that as the sample set $\mathcal{S}_t$ becomes dense, the approximate kernel $K_\mathcal{S}$ converges to the kernel using the true model, $K_\infty$, and that MCMC converges as a result. 
We begin by stating our assumptions precisely. Below $W_2$ denotes the 2-Wasserstein metric, defined in Appendix \ref{apx:proofs}.

\smallskip

\begin{defn}[Convergence Assumptions] \label{MalaAssumptions}
Assume that:
\begin{enumerate}
\item For any $\delta > 0$, there is an $\epsilon = \epsilon(\delta) > 0 $ so that any $\epsilon$-cover $\mathcal{S}$ satisfies
\be \label{EqEpsDeltCoverDef}
\sup_{\theta \in \Theta} W_{2} (K_{\mathcal{S}}(\theta,\cdot), K_{\infty}(\theta,\cdot)) < \delta. 
\ee 
\item There exist constants $0 <  \eta_{0} < \infty$ and $1 < C < \infty$ such that for any $0 < \eta < \eta_{0}$,
\be \label{IneqContAssump}
\sup_{\theta,\theta' \in \Theta, \, \| \theta - \theta' \| < \eta} W_{2} ( K_{\infty}(\theta,\cdot),  K_{\infty}(\theta',\cdot)) < C \eta.
\ee 
\item For any $\varphi_{0} < \infty$ and $\delta > 0$, there exists $\epsilon >0$ so that any $\epsilon$-covers $\mathcal{S}$, $\mathcal{S}'$ satisfy
\be  \label{IneqOtherContAssump}
\sup_{\| z \| \leq \varphi_{0}} \sup_{\theta \in \Theta} \| r_{\mathcal{S}}(\theta, z) - r_{\mathcal{S}'}(\theta, z) \| \leq \delta.
\ee 
\item Assumptions \ref{DefSomeAssumptions} hold. 
\end{enumerate}
\end{defn}
\medskip

The following theorem states that these assumptions, which we will have to check, are sufficient for convergence of the approximate Markov chain. 

\smallskip 

\begin{thm} [Convergence of Algorithm \ref{alg:algOverviewGeneral} on compact state space] \label{ThmMalaConv}
Let Assumptions \ref{MalaAssumptions} hold and let $\{X_{t}\}_{t \geq 0}$ be the sequence drawn from a run of Algorithm \ref{alg:algOverviewGeneral}. Then 
\be \label{IneqConcWassConvMala}
\lim_{t \rightarrow \infty} W_{2}(\mathcal{L}(X_{t}), \pi) = 0.
\ee 
\end{thm}

\begin{remark}
\ The proof proceeds by coupling each step of the output of Algorithm \ref{alg:algOverviewGeneral}. Our coupling construction gives us the important estimate \eqref{IneqMalaWassCoupDef}, which would not hold if the randomness at each step were resampled upon model refinement. 
In most cases, including our application to mMALA, this proof can be extended to give convergence in total variation distance by using a `one-shot' coupling (see \cite{roberts2002one}).
\end{remark}

\medskip
Finally, we observe that mMALA often satisfies Assumptions \ref{MalaAssumptions}. Although our convergence results only apply to some uses of mMALA, we believe they are representative of the more general case, and suggest the feasibility of analytically transferring the good properties of the approximation onto the kernel.

\smallskip

\begin{thm} [Convergence of approximate mMALA] \label{thm:mala}
We consider running Algorithm \ref{alg:algOverviewGeneral} with proposal kernel $q$ \change{(equivalently $r$)} given by the mMALA algorithm. Assume that:
\begin{itemize}
\item The state space $\Omega$ is the $d$-dimensional \change{hypercube} $[0,1]^{d}$ for some $d \in \mathbb{N}$. 
\item The mass matrix $M(\theta)$ and likelihood $\ell(\theta | \mathbf{d}, \mathbf{f})$ are both $C^{\infty}$ functions on $\Omega$. Furthermore, the \change{smallest singular value of $M(\theta)$ is uniformly bounded} away from zero by some $c > 0$. 
\item The posterior density $p(\cdot | \mathbf{d})$ is $C^{\infty}$ and bounded away from zero uniformly on $\Omega$.
\item  Items 3 through 6 of Assumption \ref{DefSomeAssumptions} hold.
\end{itemize}
Then  the output $\{ X_{t} \}_{t \geq 0}$ of Algorithm \ref{alg:algOverviewGeneral} satisfies 
\be 
\lim_{t \rightarrow \infty} \| \mathcal{L}(X_{t}) - \pi \|_{TV} = 0.
\ee  
\end{thm}
The proof of Theorem~\ref{thm:mala}, given in Appendix \ref{apx:proofs}, merely checks Assumptions \ref{MalaAssumptions}. Essentially, these assumptions hold because mMALA uses approximations of the derivatives of $\mathbf{f}$ to construct a Gaussian proposal; the derivative approximations improve as $\mathcal{S}$ grows and the Gaussian proposal is not too sensitive to errors in these approximations, and hence the entire kernel converges in the necessary sense.

\section{Numerical experiments}
\label{sec:numerical}

\change{We present three numerical examples to explore the algorithmic ideas developed in the preceding sections. First, we use a simple example to demonstrate how the improved mixing properties of MALA can successfully be paired with our local approximation scheme. Then, we turn to two more computationally intensive inference problems, with forward models drawn from realistic applications. The first of these, a groundwater tracer transport problem, is the focus of our parallel MCMC explorations. Though posterior evaluations are quite expensive in this problem, we can still compare results with standard MCMC chains that employ no approximation, and thus verify the accuracy of posterior expectations. The second application example is even more expensive---such that MCMC is essentially intractable without the use of approximations. Here, our goal is simply to show that with a \textit{particular instantiation} of parallel local approximation MCMC, fully Bayesian inference that previously would not have been feasible (given reasonable computational resources) is now feasible.}



\subsection{Quartic example}
\change{
\label{sec:quartic}

Consider a target distribution with the following log-quartic density: 
  \begin{equation}
    \log{\pi(x_1, x_2)} = - x_1^4 - \frac{(2 x_2-x_1^2)^2}{2},
    \label{eq:quartic_density}
  \end{equation}
also illustrated in Figure \ref{fig:quartic_density}.  We simulate from this target distribution in four ways: using (i) adaptive Metropolis (AM) \cite{Haario2001} and (ii) mMALA, each paired with either (a) evaluations of the exact target density or (b) our local approximation scheme. In other words, the combinations (a+i) and (a+ii) are standard MCMC algorithms with two different proposal schemes, and the combinations (b+i) and (b+ii) pair local approximation MCMC with the same proposal schemes. We call these simulation approaches `exact+AM,' `exact+mMALA,' `LA+AM,' and `LA+mMALA,' respectively.

 \begin{figure}[h!]
    \centering
    \includegraphics[width=0.75\textwidth]{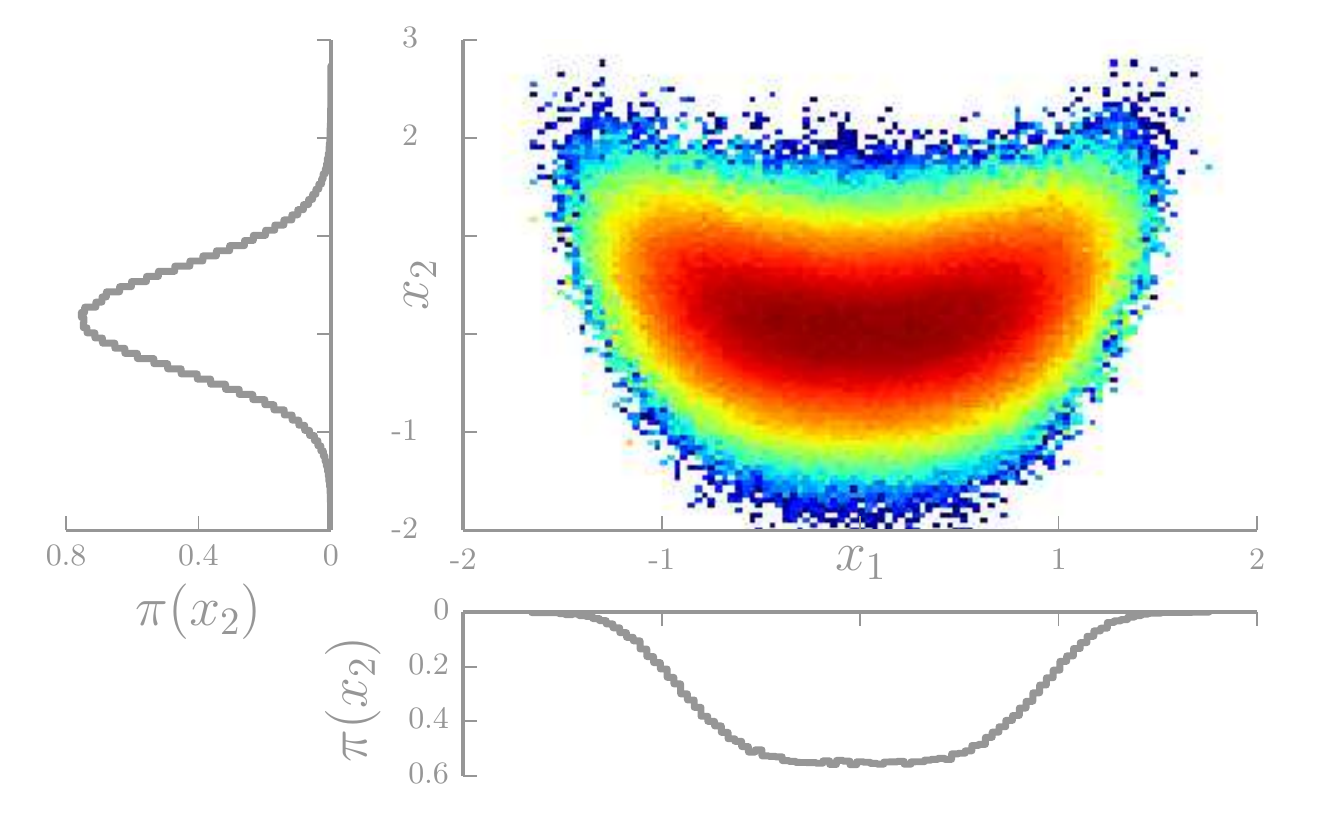}
    \caption{\change{Joint and marginal densities for the quartic target \eqref{eq:quartic_density}. We characterize this density with both exact density evaluations and local approximations, paired with adaptive Metropolis and mMALA.}}
    \label{fig:quartic_density}
  \end{figure}

For each of the simulation approaches defined above, we run 20 independent chains, each of length $6 \times 10^5$ steps. (No parallelism is employed in this example.) We then evaluate an \textit{expected squared relative error} $\bar{\varepsilon}^2$, as a function of the number of target density evaluations, for each approach. The quantity $\bar{\varepsilon}^2$ is defined as follows. Before computing expectations with respect to the target density, we discard the first $10^4$ samples of each chain as burn-in. Then we obtain a reference estimate $C_0$ for the target covariance matrix by pooling post-burn-in samples from all $40$  chains that employ exact target density evaluations. Next, for each independent chain (indexed by $i$) associated with a given simulation approach, we compute a running $t$-sample estimate of the target covariance $\hat{C}_t^{(i)}$ and define a relative squared error as
\begin{equation}
\label{eq:errorFrob}
\varepsilon_t^{2, (i)} \coloneqq \frac{ \Vert \hat{C}_t^{(i)} - C_0 \Vert_F^2}{\Vert C_0 \Vert_F^2} ,
\end{equation}
where $\Vert \cdot \Vert_F$ denotes the Frobenius norm. Then we average over the $20$ independent chains to obtain $\bar{\varepsilon}_t^2 = \frac{1}{20} \sum_{i=1}^{20} \varepsilon_t^{2, (i)}$. Figure~\ref{fig:quartic_error} plots $\bar{\varepsilon}_t^2$ versus the number of target \emph{density evaluations}, for each simulation approach. {For the exact+mMALA chains, which require direct evaluation of the gradients of $\log \pi$, we count each gradient evaluation as an additional density evaluation. For large-scale models, gradient evaluations (e.g., via an adjoint solve) might be more expensive than density evaluations, so this accounting is a conservative estimate of computational cost.}

Several trends are apparent in this figure. First, comparing the exact and local approximation chains, we see that the same level of accuracy is achieved with significantly fewer density evaluations when using approximations.  When target density evaluations are expensive, this translates to computational savings. We also note that the exact chains show a squared error decaying at roughly the standard Monte Carlo rate of $1/n$, where $n$ is the number of density evaluations. But the error decays more quickly when using local approximation MCMC.  This is because MCMC steps that do not require refinement of $\mathcal{S}_t$ can still reduce estimator variance---and thus the overall error---without using a target density evaluation.  Since we expect the refinement frequency to decay as the chain progresses, we also expect the error decay rate, in terms of the number of target density evaluations, to accelerate.

  \begin{figure}[h!]
    \centering
    \includegraphics[width=0.75\textwidth]{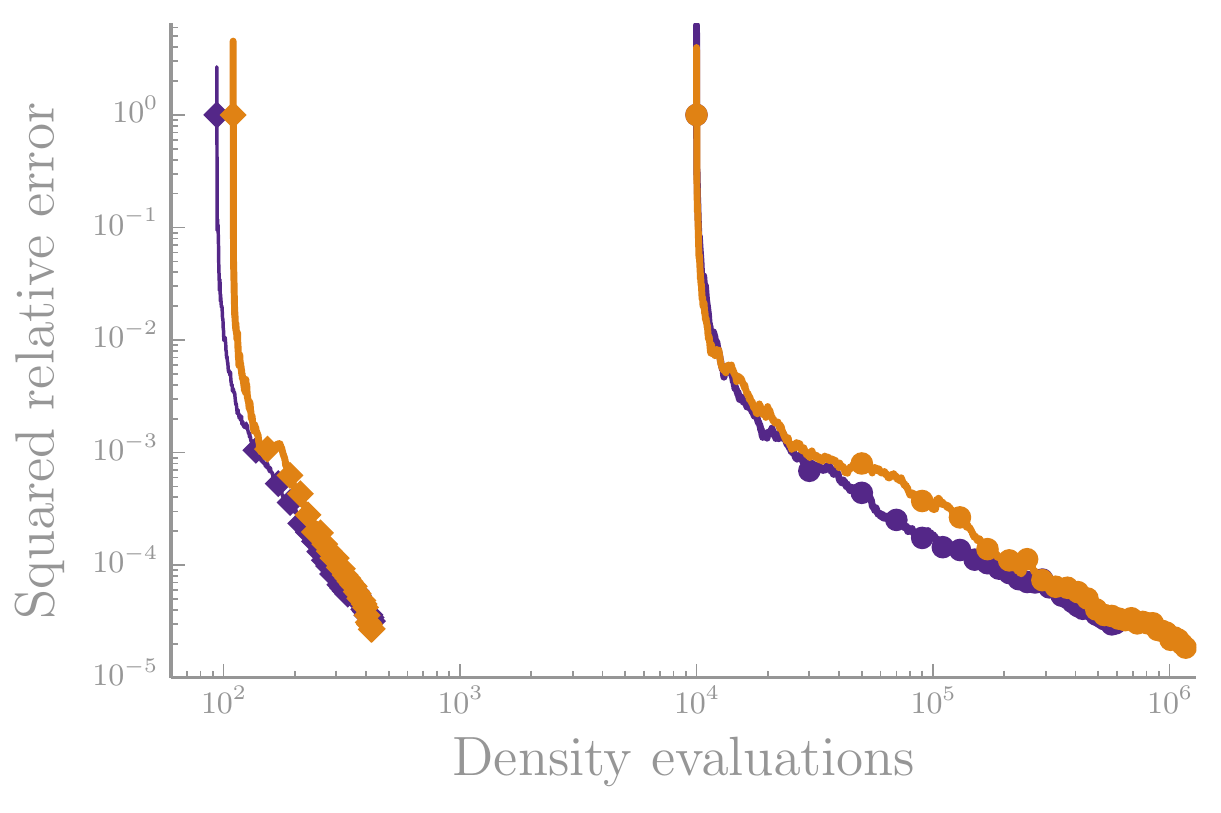}
    \caption{\change{Quartic example of Section~\ref{sec:quartic}: expected squared relative error $\bar{\varepsilon}_t^2$ as a function of the number of target density evaluations.  Purple lines correspond to  AM chains while  gold lines correspond to mMALA. The circles mark chains that employ exact evaluations of the target density, while diamonds mark chains using local approximation. In the exact case, mMALA requires evaluations of the target density and its gradient.  We assume gradient evaluations are comparable in cost to density evaluations and, therefore, count them as density evaluations.  Errors are obtained by averaging over 20 independent chains from each simulation approach, each of length $6\times10^5$ steps.}}
    \label{fig:quartic_error}
  \end{figure}

Another useful measure of sample quality is the effective sample size (ESS) of each chain, which we compute from each chain's integrated autocorrelation time \cite{Wolff04}. ESS is a measure of how many ``effectively independent'' samples have been generated from the target distribution. In Figure~\ref{fig:quartic_ess}, we plot the ESS for each independently realized chain, using each of the four simulation approaches. In general, the mMALA chains have larger ESS than the AM chains, reflecting their improved mixing for this target distribution. Also, the local approximation chains achieve nearly the same ESS as their exact counterparts, but with nearly three orders of magnitude fewer density evaluations. ESS of course varies from realization to realization; the dark symbols in the middle of the scatter plots illustrate the \textit{average} ESS and cost of each set of 20 chains. In general, we do not expect that introducing an approximation will improve mixing, and in this example ESS with exact evaluations (exact+AM or exact+mMALA) provides an upper bound on sampling performance. Indeed, Figure \ref{fig:quartic_ess} shows that the ESS is very slightly lower using local approximations; this is apparent in the MALA cases. Nonetheless, the local approximation chains achieve nearly the same ESS as their exact counterparts, but with nearly three orders of magnitude fewer density evaluations. Moreover, the improved mixing of mMALA in the exact case is preserved when using local approximations.

  \begin{figure}[h!]
    \centering
    \includegraphics[width=0.70\textwidth]{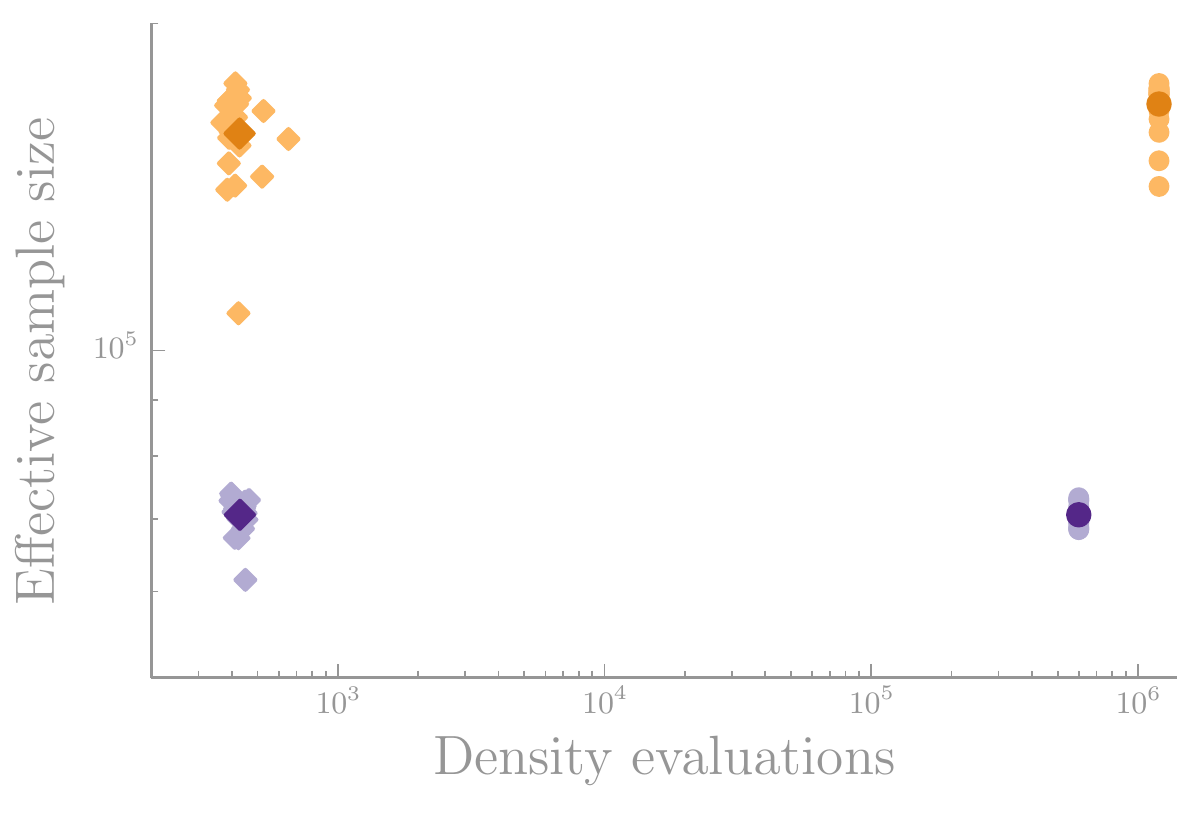}
    \caption{\change{Quartic example of Section~\ref{sec:quartic}: effective sample size for independent MCMC chains, each of length $6\times10^5$. As in Figure~\ref{fig:quartic_error},  purple symbols correspond to AM chains and gold symbols correspond to mMALA chains. Circles indicate chains using exact target density evaluations, while diamonds indicate the use of local approximations. The darker dot in each cluster is its expected value.  In the mMALA case with exact evaluations, we count target gradient evaluations as density evaluations.}}
    \label{fig:quartic_ess}
  \end{figure}
}

\subsection{Tracer transport problem}

Predicting the evolution of groundwater contaminant concentrations over time is vital to many monitoring and remediation efforts \cite{Matott2012}. A contaminant is typically modeled as a non-reactive tracer that diffuses and is advected by groundwater flow. Here we construct an inverse problem that simulates a monitoring configuration: the tracer concentration is observed at a small number of wells over a short period of time, and the subsurface conductivity field must be inferred given these data.

The conductivity field is assumed to be piecewise constant in six irregularly-shaped areas, reflecting different subsurface features (e.g., sand, clay, gravel) each with constant but unknown conductivities.  We consider a problem domain with two horizontal coordinates $x,y \in [0,1]^2$.  The true log-conductivity is depicted in Figure \ref{fig:tracer_kappa}. The conductivity is parameterized as 
$$
\kappa(x,y) = \exp{\theta_{\underline{j}(x,y)}},
$$
where $\underline{j}(x,y) \in \{1, 2, \ldots, 6\}$ is the smallest integer $j$ such that  $x_0^j \leq x \leq x_1^j$ and $y_0^j \leq y \leq y_1^j$, where the bounds $(x_0^j, x_1^j, y_0^j, y_1^j)$ are given in Table \ref{tab:rectDefs}. \change{The parameters $\theta_i$ are endowed with uniform priors; the upper and lower bounds for each prior are also given in Table \ref{tab:rectDefs}.}


\change{
  \begin{figure}
\centering
\includegraphics[width=0.85\textwidth]{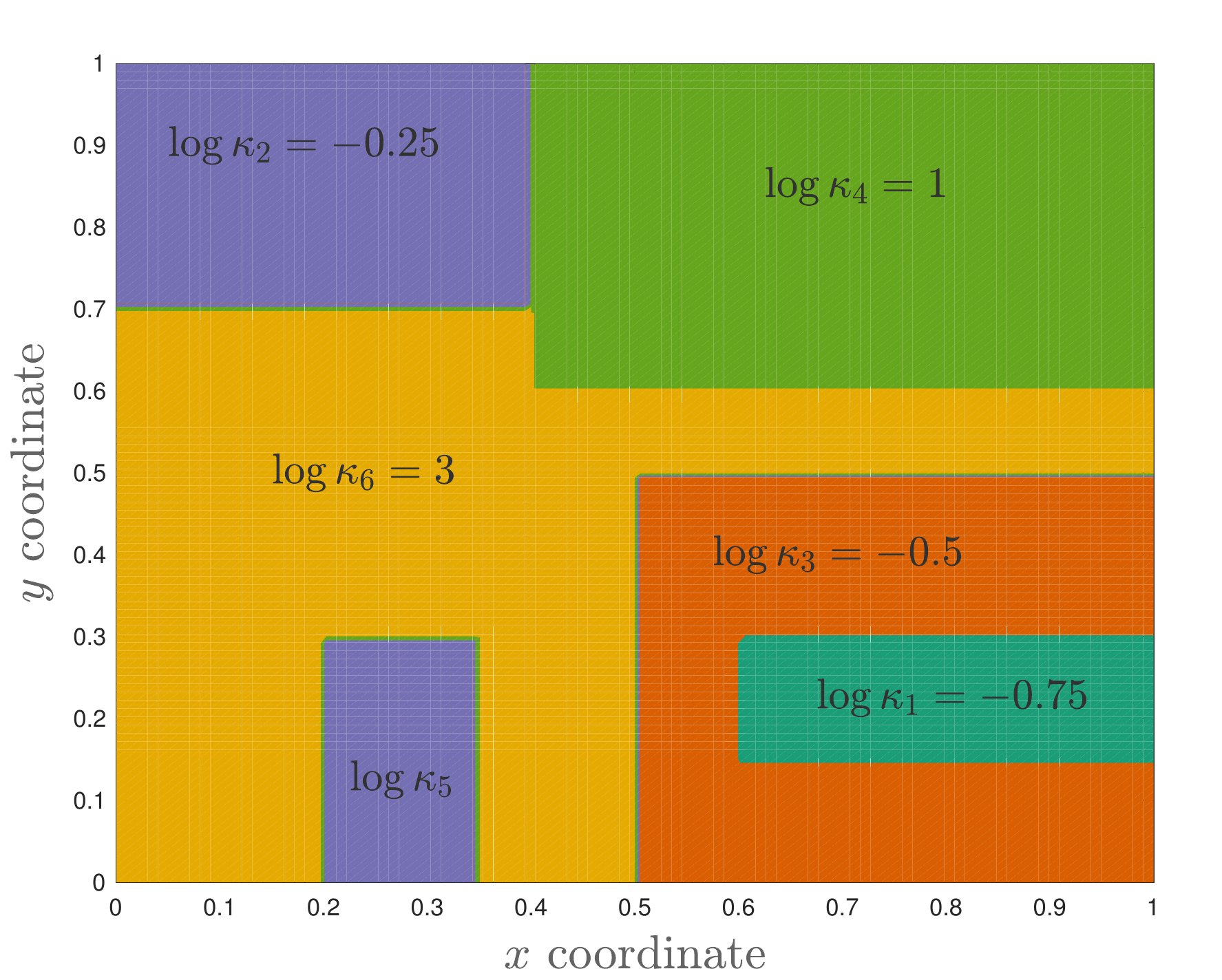}
\caption{The ``true'' log-conductivity field.}
\label{fig:tracer_kappa}
\end{figure}
 }


\begin{table}
  \change{
  \caption{\label{tab:rectDefs} True values of the parameters for the tracer problem. The log-conductivity at location $(x,y)$ is $\theta_{\underline{j}(x,y)}$, where  $\underline{j}(x,y)$ is the smallest integer $j$ such that $x_0^j \leq x \leq x_1^j$ and $y_0^j \leq y \leq y_1^j$; each parameter value $\theta_{\underline{j}}$ corresponds to $\log \kappa_{\underline{j}}$ in Figure \ref{fig:tracer_kappa}.}
  }
\ra{1.1}
\centering
\begin{tabular}{c|ccccc|cc}
\toprule
Parameter & $x_0^j$ & $x_1^j$ & $y_0^j$ & $y_1^j$ &  True value & Prior lower & Prior upper \\
\midrule
$\theta_1$ & 0.6 & 1 & 0.15 & 0.3 & -0.75 & -1 & 0 \\
$\theta_2$ & 0 & 0.4 & 0.7 & 0.1 & -0.25 & -1 & 1 \\
$\theta_3$ & 0.5 & 1 & 0 & 0.5 & -0.5 & -1 & 0 \\
$\theta_4$ & 0.4 & 1 & 0.6 & 1 & 1 & 0 & 2 \\
$\theta_5$ & 0.2 & 0.25 & 0 & 0.3 & -0.25 & -1 & 0 \\
$\theta_6$ & 0 & 1 & 0 & 1 & 3 & 2 & 5 \\
\bottomrule
\end{tabular}
\end{table}

Modeling tracer evolution requires first computing the hydraulic head, which determines the groundwater velocity.  Under the Dupuit approximation \cite{Dupuit1863}, the hydraulic head $h$ obeys the elliptic equation,
\begin{equation}
  \nabla \cdot \left( \kappa h \nabla h \right) = - f_h,
  \label{eq:hydraulicHead}
\end{equation}
where $\kappa(x,y)$ is the conductivity field and $f_h(x,y)$ is the hydraulic head forcing. In our problem setup, the forcing is created by pumping at four well locations, $(a_i, b_i) \in \{(0.15, 0.15),\, (0.85, 0.15),\, (0.85, 0.85),\, (0.15, 0.85)\}$, such that

\change{
  \be
  f_h(x,y) = \sum_{i=1}^4 p_i \exp{ \left(\frac{ (a_i - x)^2 + (b_i - y)^2}{0.02}\right)}, 
  \ee
%
where $p_i \in (10,\, 50,\, 150,\, 50)$.}  The model \eqref{eq:hydraulicHead} assumes homogeneous Dirichlet boundary conditions at $y=0$ and $y=1$ and homogeneous Neumann conditions at $x=0$ and $x=1$.  The Darcy velocity is determined by the hydraulic head gradient
\begin{equation}
  \left[ \begin{array}{c}
    u \\
    v
    \end{array} \right] = - h \kappa \nabla h.
  \label{eq:hydrologyVelocity}
\end{equation}
The time-dependent tracer concentration $c(x,y,t)$ then evolves given a flow-dependent dispersion tensor, via
\begin{equation}
  \frac{\partial c}{\partial t} + \nabla \cdot \left( \left( d_m \mathbf{I} + d_l \left[ \begin{array}{cc}
      u^2 & u v \\
      u v & v^2
    \end{array} \right] \right) \nabla c \right) - \left[ \begin{array}{c}
      u \\
      v 
      \end{array} \right] \cdot \nabla c = -f_t,
\end{equation}
where \change{$d_m=2.5\times 10^{-3}$ and $d_l=2.5\times 10^{-3}$} are dispersion coefficients and $f_t(x,y)$ is the tracer forcing.  \change{The tracer is forced by injection at each well location.  The source term is similar to the one forcing the hydraulic head
  \be
  f_t(x,y) = \sum_{i=1}^4 r_i \exp{ \left(\frac{ (a_i - x)^2 + (b_i - y)^2}{0.005}\right)}.
  \ee
where $r_i \in (10,\, 5,\, 10,\, 5 )$.}
The tracer has initial condition $c(x,y,0)=0$, and homogeneous Neumann conditions are enforced at all spatial boundaries. Since the hydraulic head forcing, tracer forcing, and dispersion coefficients are known, the forward model simply maps the conductivity to a time-evolving concentration field.  
Tracer observations are taken at \change{25 well locations: $(x_i, y_j)$ such that $x_i = 0.1 + \frac{i-1}{5}$ and $y_j = 0.1 + \frac{j-1}{5}$ for $i,\, j \in \{1,\, \hdots,\, 5\}$ at successive times $t \in \{0.1,\, 0.2,\, 0.3,\, 0.4,\, 0.5\}$.}

\change{
  \begin{figure}
    \centering
    \includegraphics[width=0.85\textwidth]{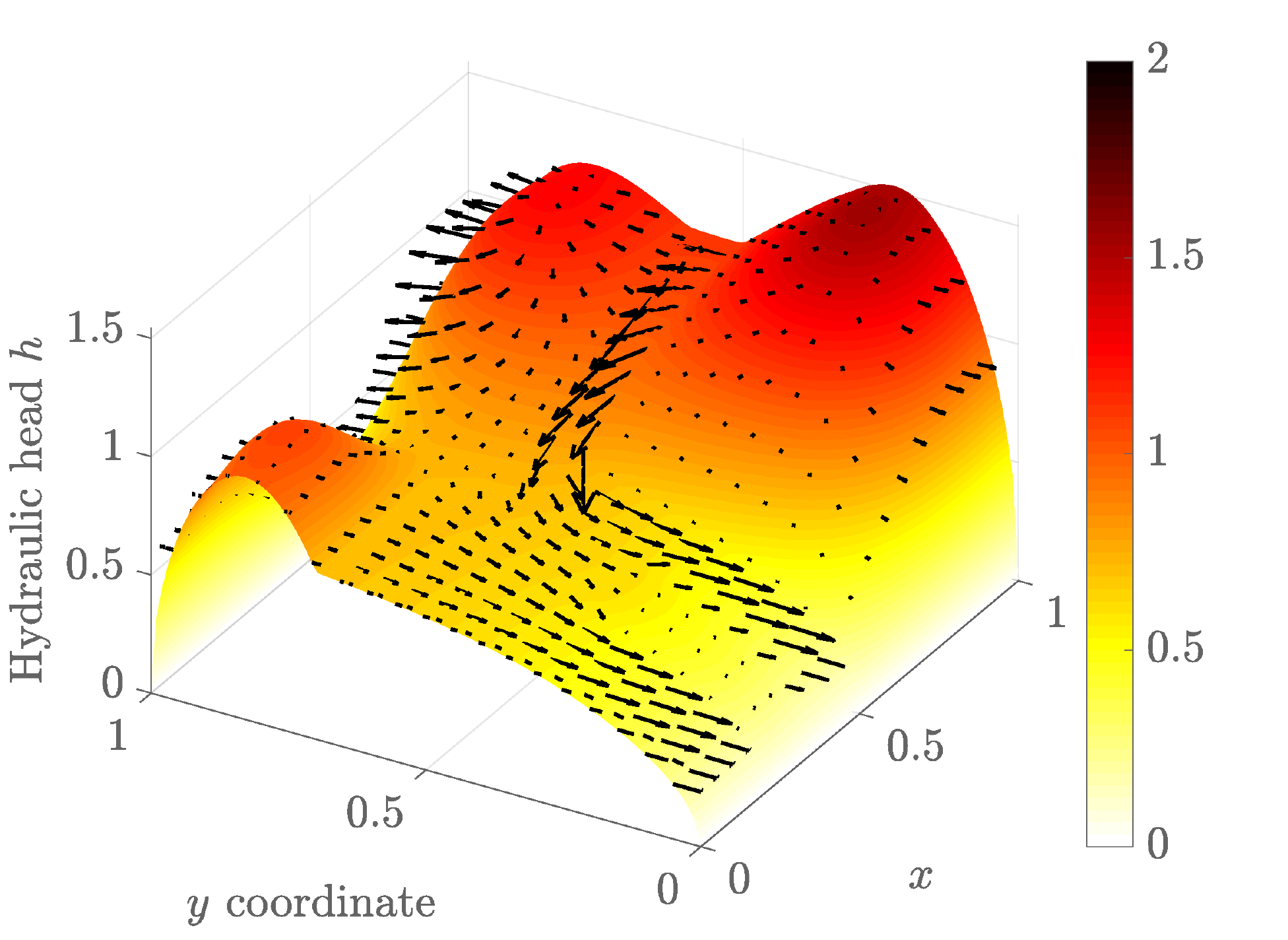}
    \caption{Hydraulic head (colormap) $h(x,y)$ computed via (\ref{eq:hydraulicHead}) and corresponding velocities (\ref{eq:hydrologyVelocity}) (arrows), given the conductivity field in Figure \ref{fig:tracer_kappa}.}
    \label{fig:tracer_pressure}
  \end{figure}
}


\change{
  \begin{figure}
    \centering
    \includegraphics[width=0.85\textwidth]{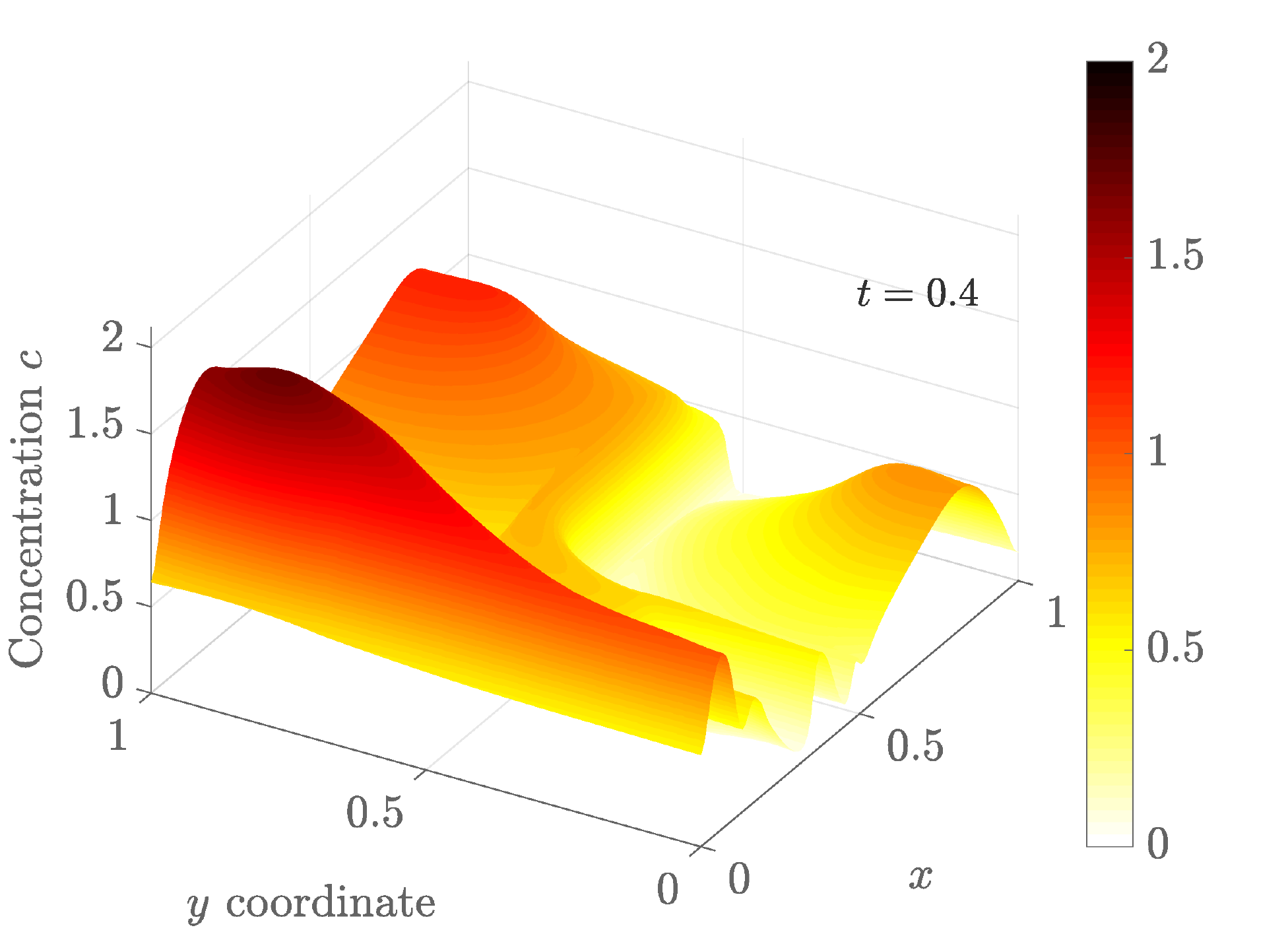}
    \caption{The tracer concentration $c(x,y,t=0.4)$, given the conductivity field in Figure \ref{fig:tracer_kappa}. The tracer is injected from a well in each corner.}
    \label{fig:tracer_observations}
  \end{figure}
}


The forward solver computes the steady state pressure and velocity fields, then simulates the tracer advection/diffusion. \change{Figure \ref{fig:tracer_pressure} shows the hydraulic head and velocity fields resulting from the true log-conductivity, and Figure \ref{fig:tracer_observations} shows the associated tracer concentration field at $t=0.4$. Overall, the parameter-to-observable map, from the log-conductivities to the time-dependent tracer concentrations, is strongly nonlinear and challenging to approximate.}
Data for inversion are generated using a standard finite element scheme on \change{a $200 \times 200$ mesh.} The solver used for inversion (i.e., to evaluate the posterior density at a candidate value of $\theta$) uses a coarser \change{$100 \times 100$ mesh.}  \change{In both cases (generating the data and within the inversion), time integration of the contaminant concentration field uses a Crank-Nicolson scheme.}  The likelihood assumes additive and i.i.d.\ errors for each observation of tracer concentration, Gaussian with mean zero and variance \change{$10^{-2}$}.  

\change{In a serial implementation, each evaluation of the forward model and hence the likelihood requires roughly 13 seconds of computation. Though we will mitigate this cost using local approximations, we also wish to compare our approach with chains that employ exact evaluations of the forward model. To make such comparisons feasible---and also to reflect computational practice for complex PDE models---we parallelize each forward model evaluation. We use four processors, which reduces the forward model's runtime to roughly 4 seconds of computation. Thus our parallel MCMC scheme actually employs two levels of parallelism: an outer level involving parallel chains, as described in Section~\ref{sec:parallel}, and an inner level within each forward model evaluation.}


\change{
  \begin{figure}[htb]
    \centering
    \subfloat[exact+AM]{
      \def\svgwidth{0.4\textwidth}
      \includegraphics[width=0.45\textwidth]{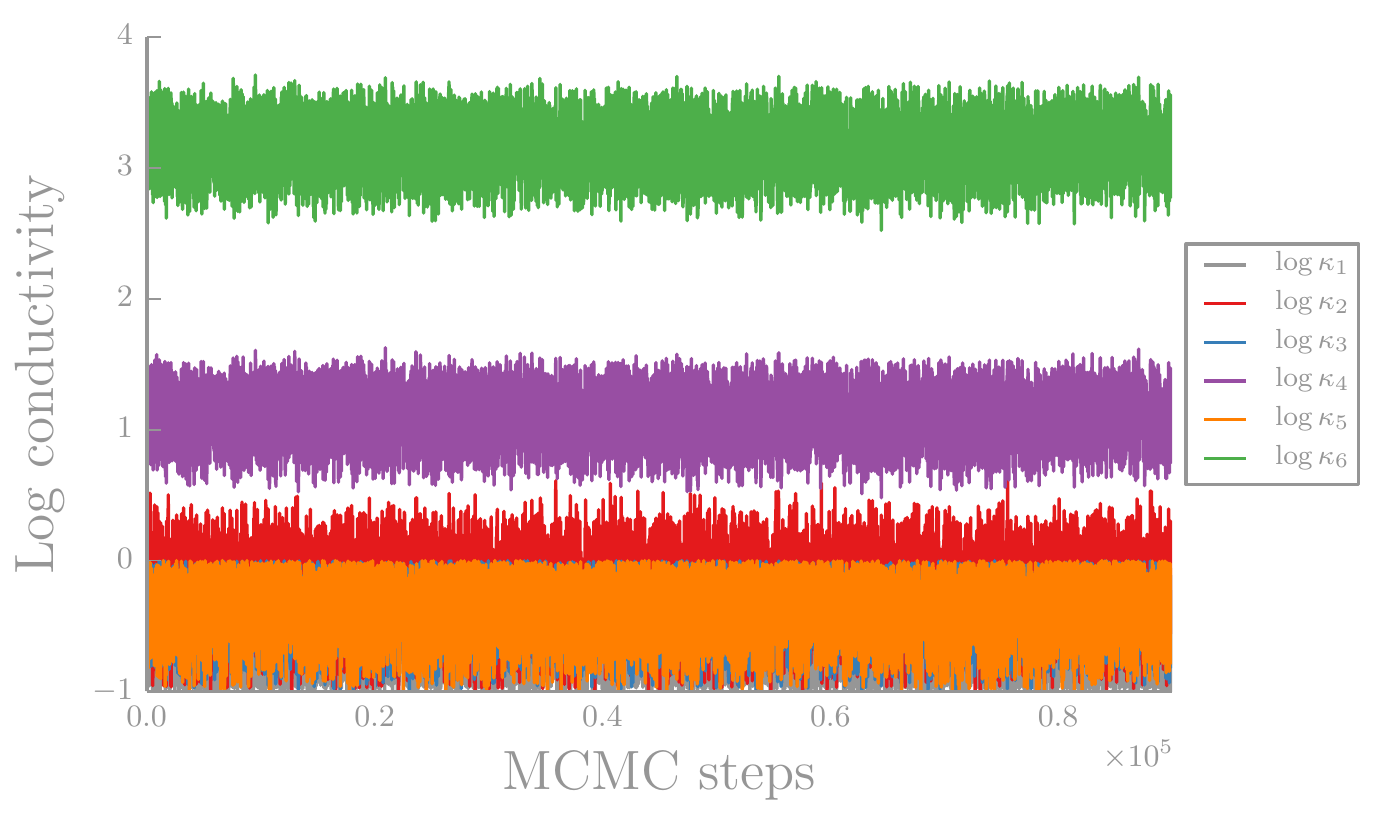}
    \label{fig:tracer_mixing_exact_AM}}
    \subfloat[LA+AM]{
      \def\svgwidth{0.4\textwidth}
      \includegraphics[width=0.45\textwidth]{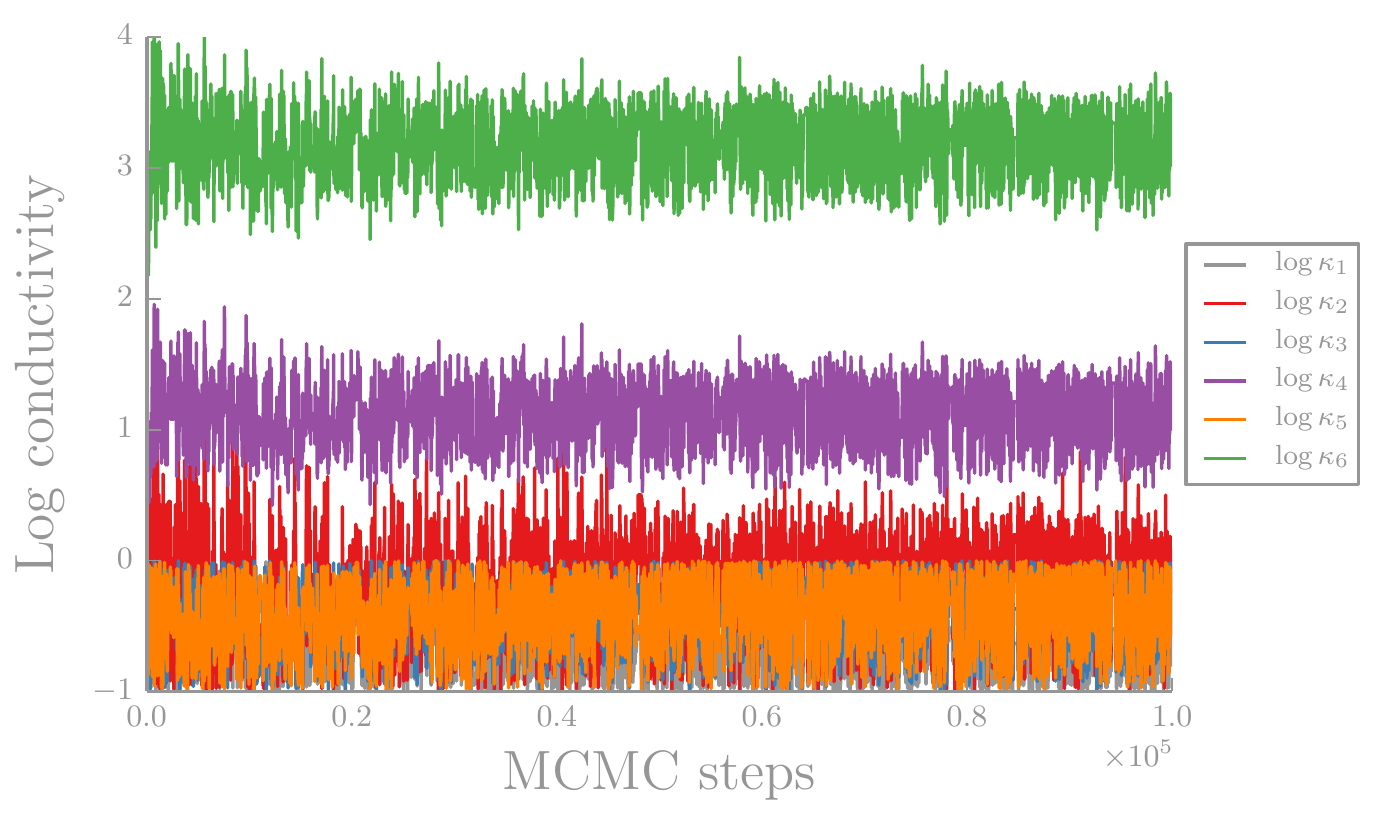}
      \label{fig:tracer_mixing_LA_AM}}\\
    \subfloat[exact+MALA]{
      \def\svgwidth{0.4\textwidth}
      \includegraphics[width=0.45\textwidth]{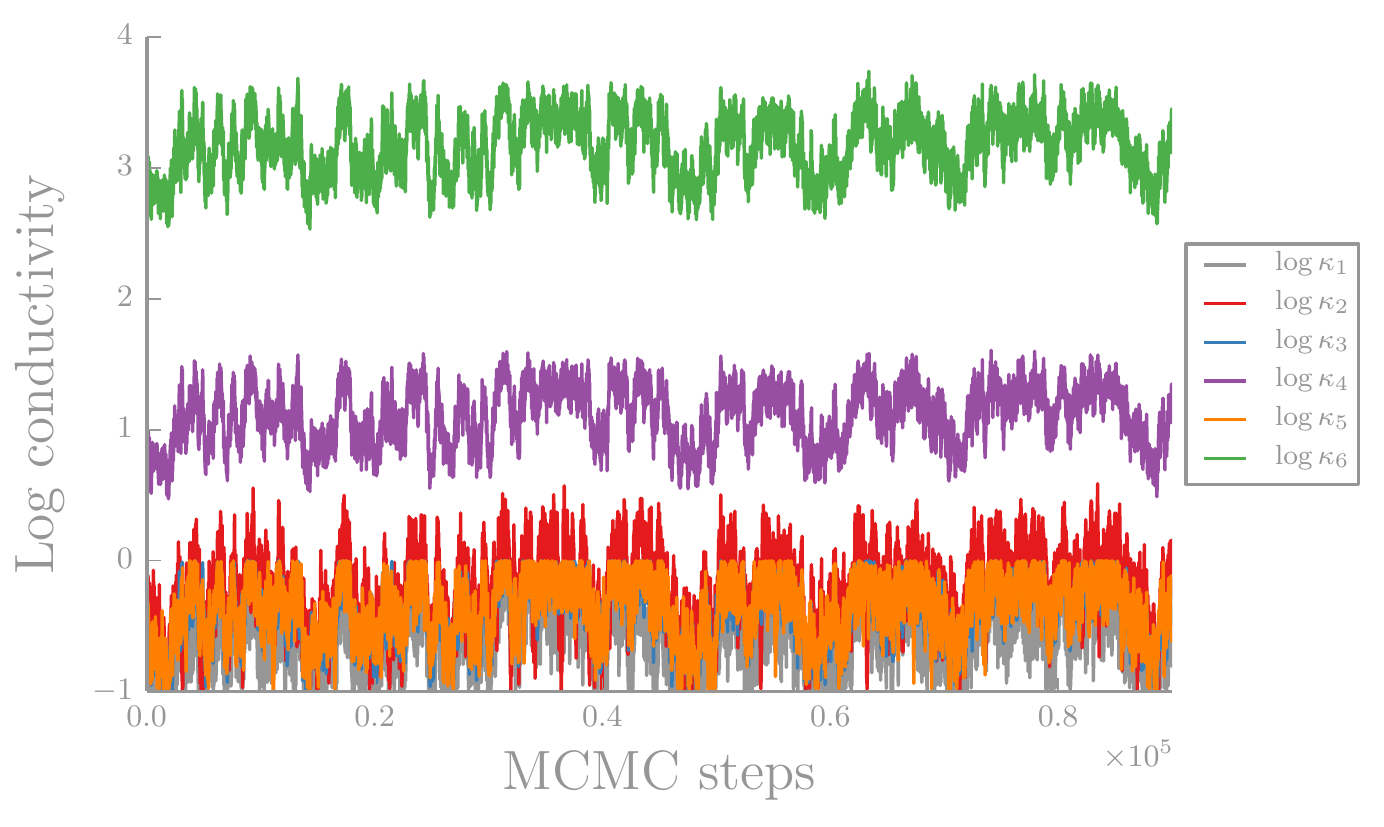}
      \label{fig:tracer_mixing_exact_MALA}}
    \subfloat[LA+MALA]{
      \def\svgwidth{0.4\textwidth}
      \includegraphics[width=0.45\textwidth]{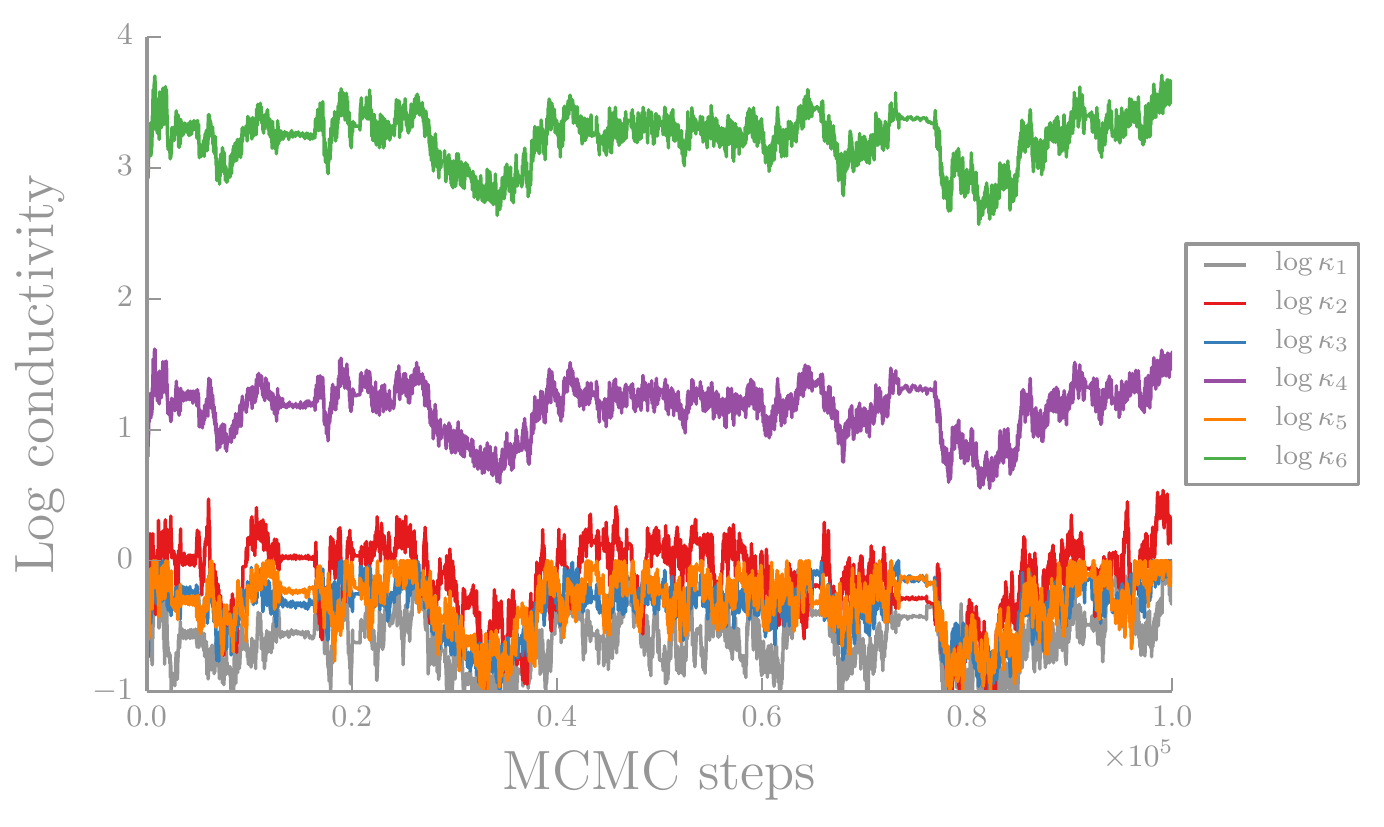}
      \label{fig:tracer_mixing_LA_MALA}}
    \caption{\change{Tracer transport problem: trace plots for a single MCMC chain (state versus MCMC iteration) either using exact density evaluations or employing a local approximation (LA), paired with either an AM or mMALA proposal.}}
    \label{fig:tracer_mixing}
  \end{figure}
}

The posterior distribution in this problem has no standard analytical form. To establish a baseline for accuracy comparisons, we instead run \change{31} independent exact+AM chains. Each chain is $10^5$ steps long, which requires several days (per chain) of computation. After discarding the first $10^4$ samples of each chain as burn-in, the remaining samples are pooled and used to characterize the posterior distribution. \change{Figure \ref{fig:tracer_mixing_exact_AM} shows a trace plot of one such exact+AM chain, for all six components of the state.  Visually, the transient behavior of the chain appears exhausted well before ${10^4}$ steps, justifying our choice of burn-in.}   One- and two-dimensional marginals of the posterior distribution, computed using the pooled exact+AM chains, are shown in Figure \ref{fig:tracer_density}. The distribution has distinctly non-Gaussian structures, \change{and the regions of high posterior probability seem to concentrate around the ``true'' parameters given in Table \ref{tab:rectDefs}.}

\change{While the AM chains appear to mix well for this problem, mMALA proves far less effective. Figure \ref{fig:tracer_mixing_exact_MALA} shows trace plots of an exact+mMALA chain targeting the same posterior. This calculation is rather laborious (over $415$ hours), as direct evaluations of the gradient of the forward model are not available; instead we compute the gradients using finite differences. This simulation is not intended as a practical approach, but rather to assess the performance of mMALA in the absence of local approximations. We find that the chain mixes quite poorly; the ESS after $10^5$ MCMC steps is only 80. Based on the results of Section~\ref{sec:quartic}, we do not expect mMALA paired with local approximations to fare any better and, indeed, Figure \ref{fig:tracer_mixing_LA_MALA} shows that mixing is poor for an LA+mMALA chain. Given these results, we focus the rest of this section on AM chains, with a goal of exploring the performance of parallel LA schemes. 
More broadly, we note that there is no guarantee that MALA schemes should improve over adaptive Metropolis (or even simple random-walk Metropolis) in low-dimensional problems such as those considered here. The potential for such improvements is problem-dependent and sometimes rather delicate, as was recognized almost immediately when MALA was introduced \cite{roberts1998optimal}. 

We first examine the convergence of estimates produced by \textit{single} LA+AM chains. Algorithm settings are given in Appendix~\ref{apx:code}, and code for this example is provided in the Supplementary Material. We run 51 independent chains, again discarding the first $10^4$ samples of each chain as burn-in. For consistency, we simply choose the same burn-in period for the exact chains and LA chains. If anything, this choice is less favorable to LA---though asymptotically it is immaterial. The mixing of a single LA+AM chain is visualized by the trace plot in Figure \ref{fig:tracer_mixing_LA_AM}. Initially, the chain does not mix as quickly as in the exact+AM case, but mixing improves as the approximation is refined, and overall the chain appears to explore the posterior quite efficiently. We also emphasize that the horizontal axis in Figure \ref{fig:tracer_mixing_LA_AM} does not reflect computational cost, since the latter is dominated by target density evaluations rather than MCMC steps. 

To assess error versus computational cost,  Figure \ref{fig:tracer_accuracyTime_single} shows, for each individual chain, the squared relative error in a running posterior covariance estimate versus wall clock time. The squared relative error $\varepsilon_t^{2, (i)}$ is defined in \eqref{eq:errorFrob}, where the reference value  $C_0$ of the posterior covariance is computed by pooling all $2.79\times 10^6$ available exact+AM samples. For comparison, we also plot error versus run time for 31 exact+AM chains. When reporting wall clock times here and below, we include the computational cost of the entire chain, including the cost of portions discarded as burn-in.
Error in the LA chains decreases steadily and reaches an accuracy comparable to the exact chains, but with significantly shorter run times. We also notice that decay rate of the expected error (bold red line in Figure \ref{fig:tracer_accuracyTime_single}) in the LA case seems to accelerate. As noted in the quartic example (where longer chains accentuated this trend), this acceleration is due to the fact that refinements happen less frequently as the chain progresses, while additional MCMC steps continue to reduce the error.  
}

\begin{figure}
\centering
\includegraphics[width=.8\textwidth]{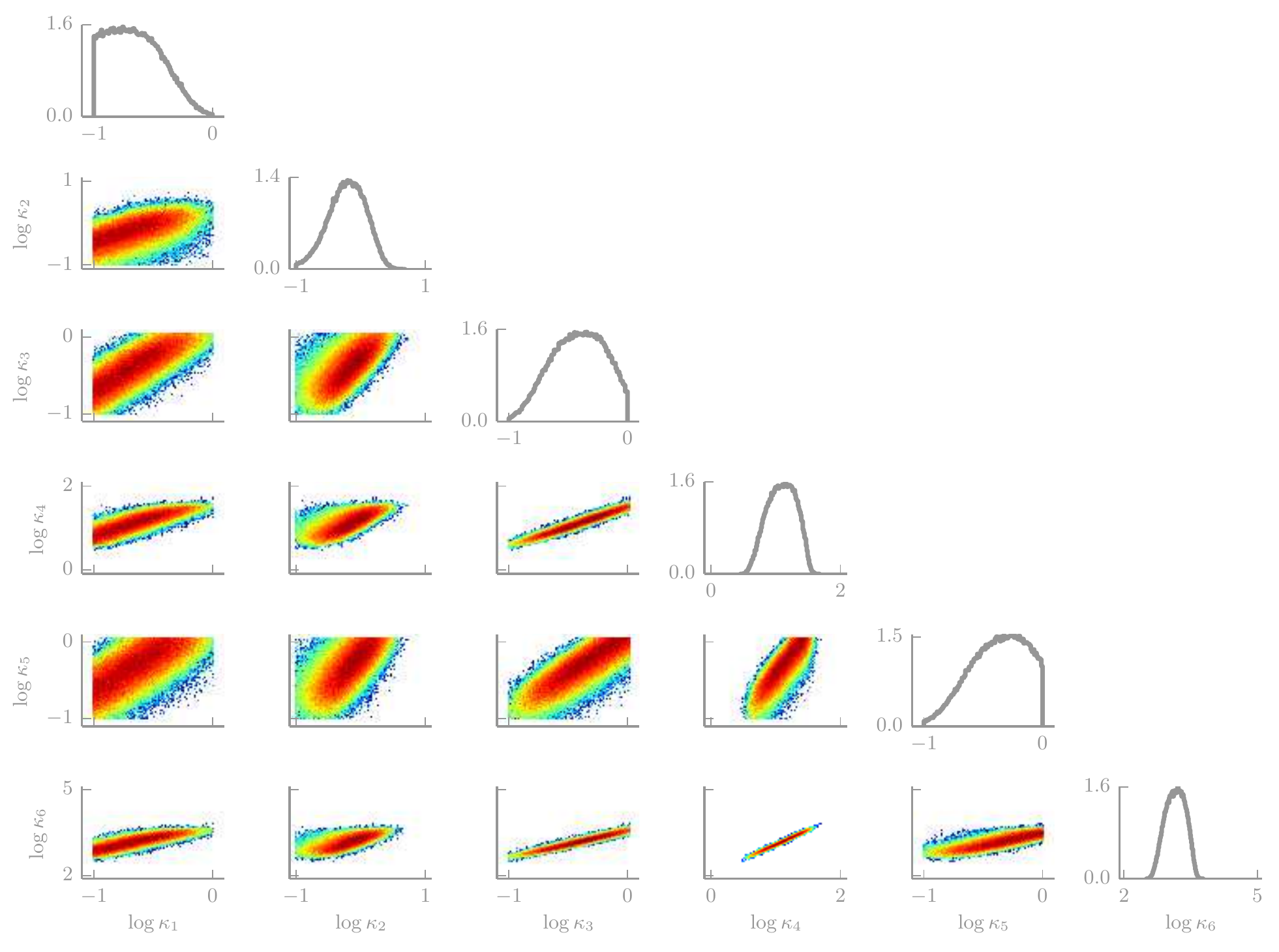}
\caption{One- and two-dimensional posterior marginals of the parameters in the hydrologic tracer transport problem.  \change{Bounds on each subplot axis are the upper and lower bounds for the uniform prior on the corresponding parameter (Table \ref{tab:rectDefs}).}}
\label{fig:tracer_density}
\end{figure}


\begin{figure}
\centering
\includegraphics[width=.8\textwidth]{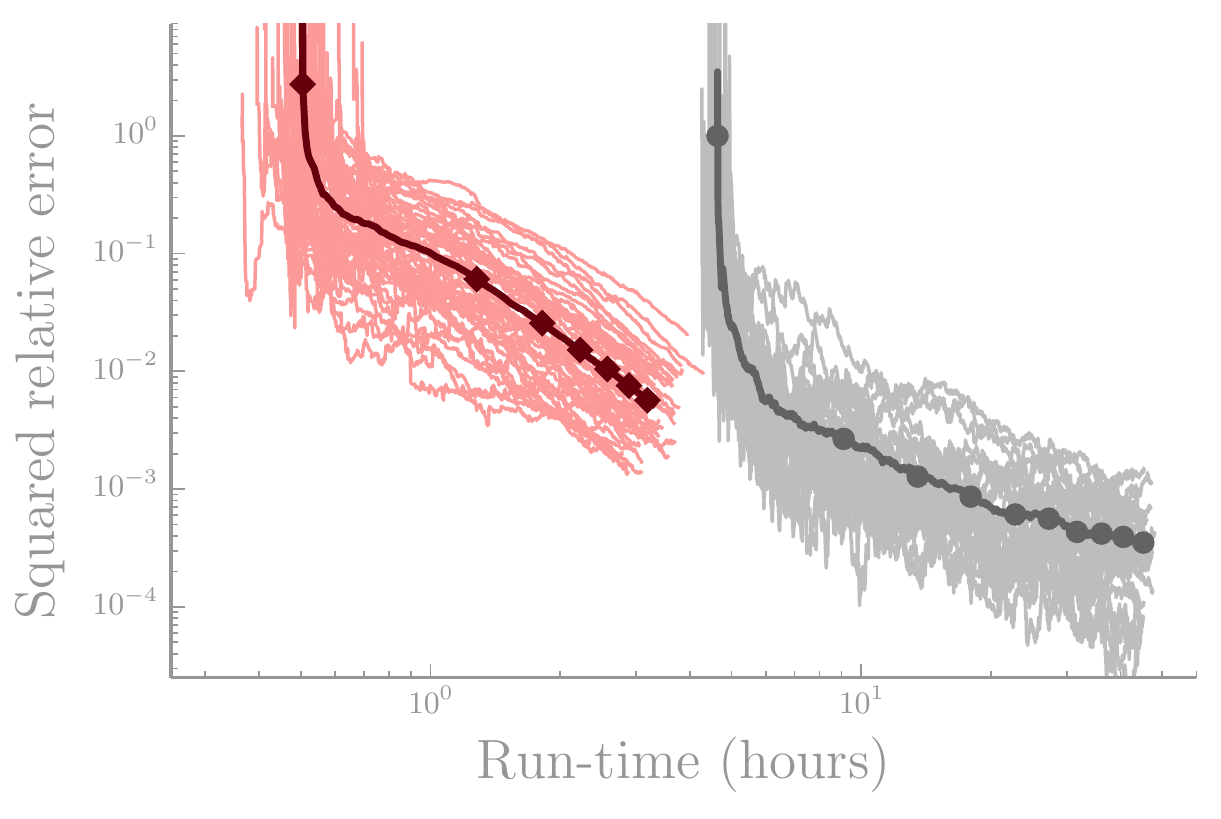}
\caption{\change{Tracer transport problem: relative squared error in the posterior covariance estimates produced by independent single (i.e., not parallel) AM chains, versus run time. The light gray lines correspond to 31 independent exact+AM chains, each of length $10^5$.  The dark gray line shows the expected error for this exact case.  The light red lines correspond to 51 single LA+AM chains, each of length $2 \times 10^5$. The dark red line shows the expected error in the approximate case.}}
\label{fig:tracer_accuracyTime_single}
\end{figure}

\change{
The local approximation sampler becomes even more effective in a parallel chain setting, where concurrent chains are allowed to {share} posterior density evaluations by building a common $\mathcal{S}_t$. The colored lines in Figure \ref{fig:tracer_accuracyTime_Parallel} show error versus run time for increasing levels of parallelism $k$, from 1 to 30 chains. To assess the variability of the error, each $k$-chain simulation is repeated several times; each such realization is shown on the figure. Each individual LA+AM chain (within a group of $k$) has a fixed length of $10^5$ steps and, as before, the first $10^4$ samples of each chain are discarded as burn-in. The error plotted on the vertical axis is again the squared relative error in the posterior covariance. Two trends are visible in the colored lines. First, as the number of chains increases, the error decreases. In and of itself, this is not surprising: summing across the chains, we accumulate more MCMC samples and, along the way, seek more model evaluations to refine the local approximations (this will be quantified precisely in subsequent figures). But the colored lines \textit{also} move to the left as the number of parallel chains increases; in other words, both the error \textit{and} the run time are reduced. This trend contrasts with that obtained by simply running exact+AM chains in parallel, an exercise depicted by the gray lines in Figure \ref{fig:tracer_accuracyTime_Parallel}. Using this na\"{i}ve parallelization, adding more chains decreases the sampling error but does not affect the run time. Moreover, the run times of LA+AM are one to two orders of magnitude smaller for comparable errors. 

}

\begin{figure}
\centering
\includegraphics[width=.8\textwidth]{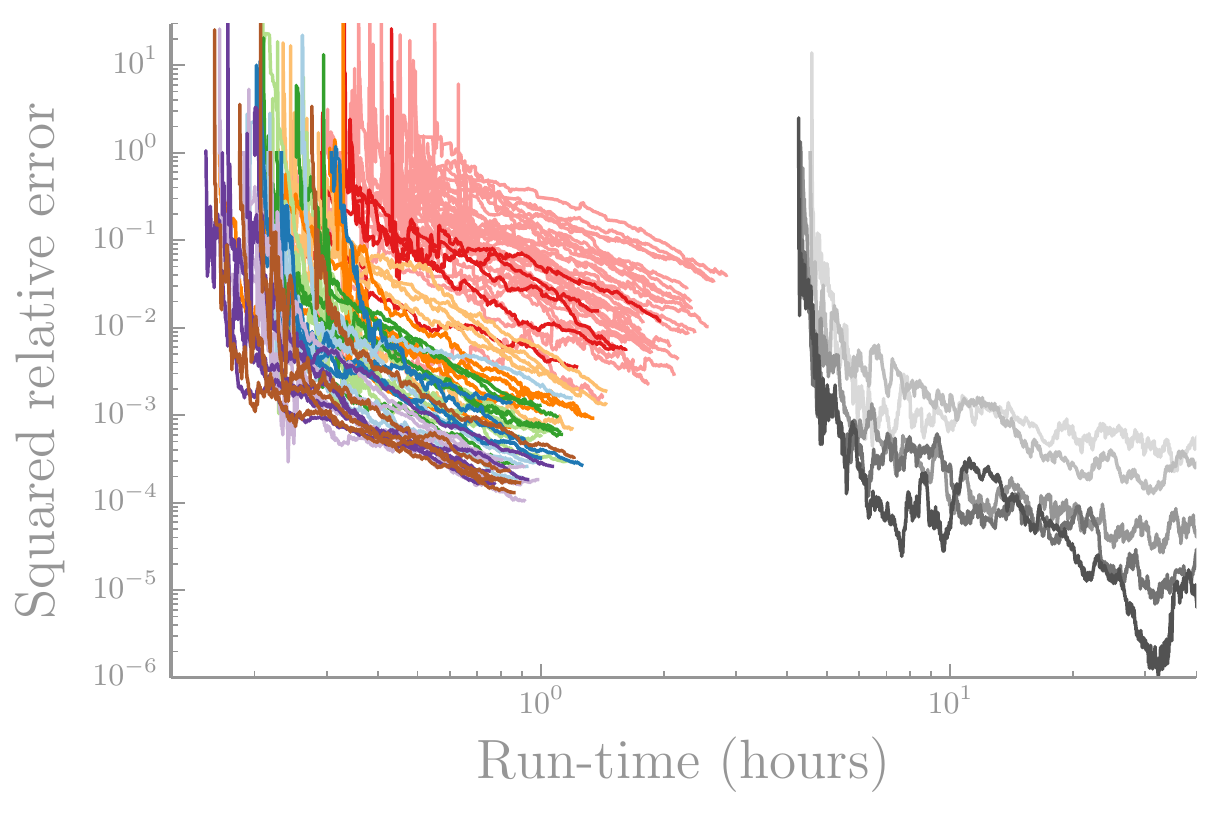}
\caption{\change{Tracer transport problem: relative squared error in the posterior covariance estimates obtained from {parallel} MCMC chains. The gray lines are computed using exact target density evaluations for $k \in \{1,\,2,\,4,\,8, \mbox{ or } 16\}$ chains.  Darker shades correspond to simulations with more parallel chains. The colored lines are computed using local approximation MCMC.  We use $k \in \{1,\,2,\,4,\,6,\,8,\,10,\,13,\,16,\,20,\,25,\mbox{or }30\}$ chains corresponding to light red, red, light orange, orange, light green, green, light blue, blue, light purple, purple, and brown, respectively. The error is that of a running covariance estimate obtained by pooling samples from the $k$ concurrent chains.  Sharing posterior density evaluations shortens the runtime \emph{and} reduces the error.}}
\label{fig:tracer_accuracyTime_Parallel}
\end{figure}



\change{We can also characterize the behavior of parallel local approximations by evaluating ESS as a function of computational effort. Figure \ref{fig:tracer_essTime} shows ESS as a function of wall clock time. First, as a baseline, consider again running exact+AM chains of length $10^5$ in parallel, depicted by gray and black circles. We certainly expect parallel chains to yield a larger ESS once their samples are pooled, and indeed the circles jump upwards as we increase the number of concurrent chains from 1 to 30. Increasing the number of chains in the exact case does not, however, change the time it takes to simulate each chain; thus the gray and black dots are vertically aligned at the same run times. In the parallel LA+AM cases, depicted by colored diamonds, the story is more interesting. As the number of parallel chains increases, the symbols move upwards \emph{and} to the left, reflecting decreased run times. Several independent realizations of each parallel case are presented, since the simulations are not deterministic. Note that the ESS of a single LA+AM chain (light red) is lower than that of an exact+AM chain of the same length; this is expected, given the mixing comparison at the top of Figure \ref{fig:tracer_mixing}. Similarly, $30$ parallel exact+AM chains have a higher combined ESS than $30$ parallel LA+AM chains (the brown diamonds of Figure \ref{fig:tracer_essTime}). But the latter entail a vastly smaller computational effort. Because of the collaboration among chains, we can compute a larger number of independent samples in less time.

Our second comparison uses a more stringent measure of parallel efficiency: ESS per chain--hour, i.e., the total ESS divided by the number of chains and the wall clock time.  This measure removes the intrinsic advantage of having multiple chains.  A na\"{i}ve MCMC parallelization yields no improvement in efficiency according to this metric: the number of independent samples might grow linearly with the number of chains, but this growth is normalized away. Figure \ref{fig:tracer_essTimeRel} shows this behavior for exact+AM chains using gray circles. In contrast, the results of parallel local approximation, depicted by colored diamonds, show steady gains in ESS/(\change{chain--hour}) with additional parallel chains.} This gain is the result of collaboration among the chains in the most computationally expensive element of the inference problem---evaluating the posterior density---by sharing evaluations from which we construct a shared surrogate model.  We note that the total number of model evaluations performed during the parallel experiments is still higher than in a single-chain case, but since the additional evaluations are parallelized, the run time is shorter.

\begin{figure}
\centering
\includegraphics[width=.8\textwidth]{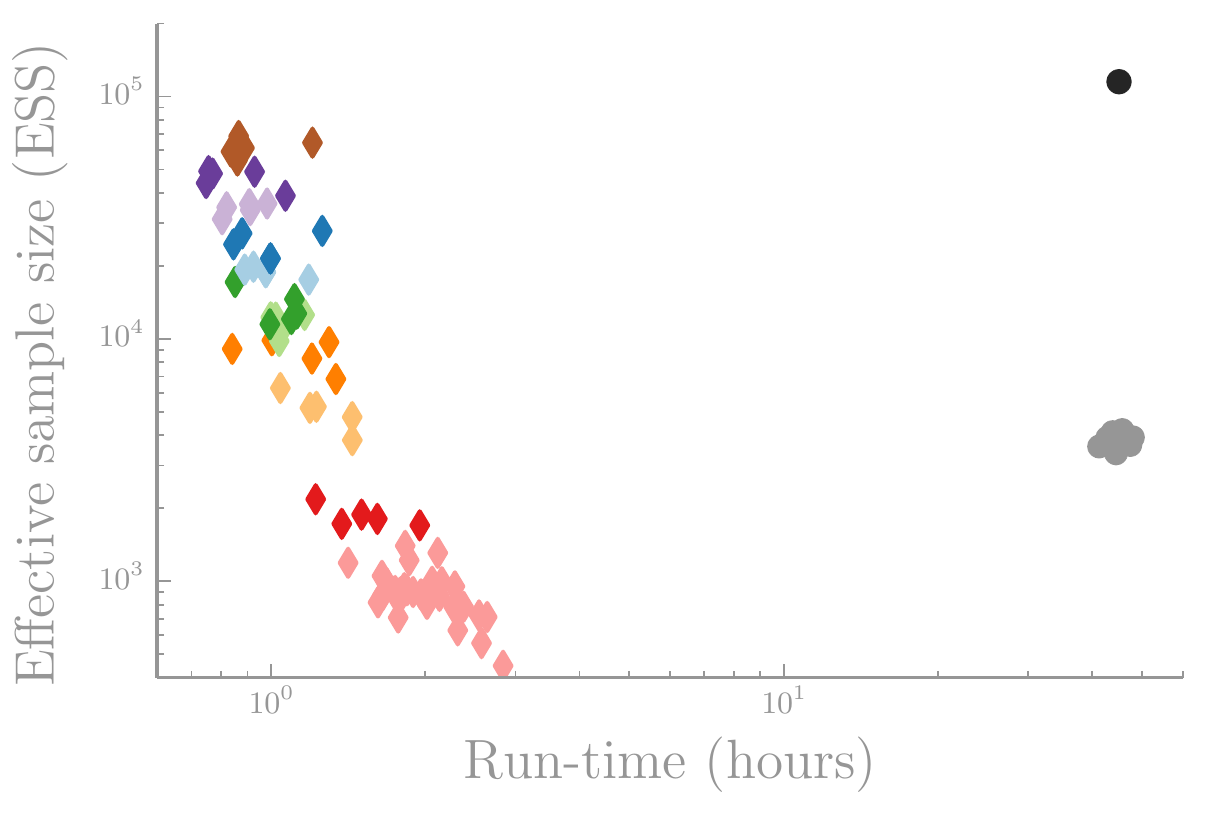}
\caption{\change{Results of the parallel efficiency study on the tracer transport problem, comparing run time to the total ESS across parallel chains. Each symbol represents one (parallel) experiment.  Light gray circles and light red diamonds correspond to single chains of length $10^5$ (with $10^4$ burn-in) using exact evaluations and local approximation, respectively.  Each colored diamond represents a different number $k$ of parallel LA+AM chains, $k \in \{1,\,2,\,4,\,6,\,8,\,10,\,13,\,16,\,20,\,25,\mbox{or }30\}$, and the colors are as in Figure~\ref{fig:tracer_accuracyTime_Parallel}.  The black circle corresponds to 30 parallel exact+AM chains. Using local approximations, running more chains increases ESS \emph{and} decreases run time.}}
\label{fig:tracer_essTime}
\end{figure}



\begin{figure}
\centering
\includegraphics[width=.8\textwidth]{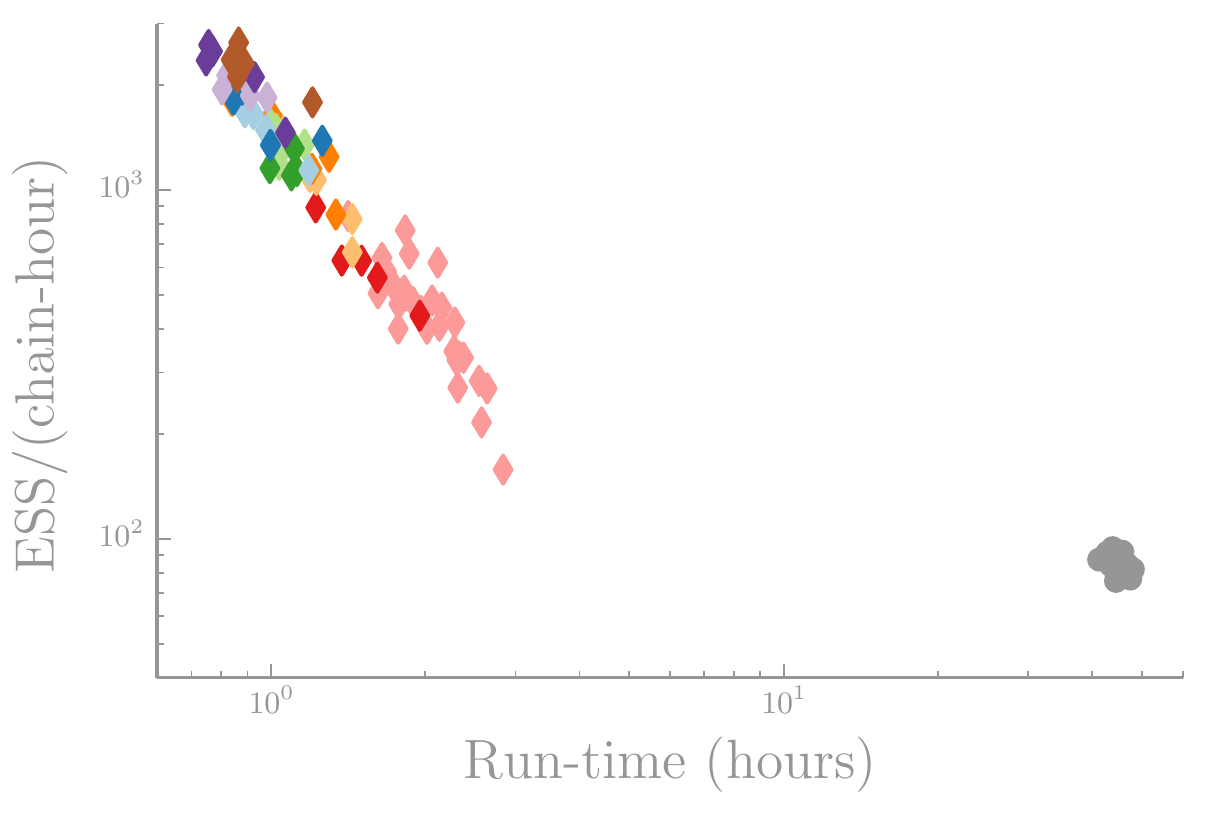}
\caption{Results of the parallel efficiency study on the tracer transport problem, comparing run time to the effective number of samples produced per chain--hour. Each symbol represents one (parallel) experiment. Again, colors from light red to brown correspond to more parallel chains, $k \in \{1,\,2,\,4,\,6,\,8,\,10,\,13,\,16,\,20,\,25,\mbox{or }30\}$.  Using parallel local approximations, ESS per chain--hour increases with the number of chains.}
\label{fig:tracer_essTimeRel}
\end{figure}



\subsection{Shallow-shelf ice stream model}

Continental ice sheets are divided into basins that are drained by fast-flowing river-like ice streams. These ice streams regulate the discharge of ice mass into the ocean, and hence play a key role determining the overall behavior of the ice sheet. The IPCC has identified the Antarctic contribution to sea-level rise as an important source of uncertainty in climate projections, and ice streams have become a widespread topic of study \cite{IPCC2014Synthesis, IPCC2014Chapter13}. 

Ice stream dynamics are not completely understood, nor are the factors governing their dynamics. Although satellite data provide plentiful observations of topology and surface velocities \cite{bedmap2, IceBridge, ICESat}, basal properties, such as the friction between the base of the ice and the underlying ground---the \emph{basal friction}---are difficult or impossible to observe directly. The basal friction varies widely, and may be higher if the ice is scraping directly against rough bedrock or lower if the ice rests on till, a mixture of mud and rock that lubricates the interface. The basal friction also parameterizes basal lubrication caused by melting basal ice (possibly due to geothermal or frictional heating).  Previous work infers basal friction given surface velocity observations \cite{MacAyeal1993, Petraetal2012}; quantifying uncertainty in the basal friction, however, requires considerable computational expense and/or posterior approximations \cite{Petraetal2014}.  In this example, we explore the problem of inferring the basal friction from surface velocities, employing local approximations to reduce the computational cost of MCMC.

Ice is often modeled as a highly, viscous non-Newtonian, and incompressible fluid. In particular, the shallow-shelf approximation \cite{MacAyeal1989, MacAyeal1993, MacAyeal1997} describes ice stream velocity assuming that \textit{(i)} the horizontal extent ($\mathcal{O}(100 \, \mbox{km})$) is much larger than the vertical extent ($\mathcal{O}(1 \, \mbox{km})$); and \textit{(ii)} the vertical velocity is zero.  The nondimensionalized shallow shelf equations for a two-dimensional horizontal domain $[0,1]^2 \ni (x,y)$ are
\begin{eqnarray*}
  \frac{\partial}{\partial x}\left( 2 \nu h \left( 2 \frac{\partial u}{\partial x} + \frac{\partial v}{\partial y} \right) \right) + \frac{\partial}{\partial y}\left( \nu h \left( \frac{\partial u}{\partial y} + \frac{\partial v}{\partial x} \right) \right) - \beta \left|u\right|^{m-1} u = h \frac{\partial s}{\partial x} \\
  \frac{\partial}{\partial x}\left( 2 \nu h \left( 2 \frac{\partial v}{\partial y} + \frac{\partial u}{\partial x} \right) \right) + \frac{\partial}{\partial y}\left( \nu h \left( \frac{\partial v}{\partial x} + \frac{\partial u}{\partial y} \right) \right) - \beta \left|v\right|^{m-1} v = h \frac{\partial s}{\partial y},
\end{eqnarray*}
with boundary conditions 
\begin{gather*}
  \begin{array}{ccccccc}
    u = 0 & \mbox{and} & v = -1 & \mbox{at} & x = 0 & \mbox{and} & x = 1,
  \end{array} \\
  \begin{array}{ccccccccc}
    \frac{\partial u}{\partial y} + \frac{\partial v}{\partial x} = 0 & \mbox{at} & y=0 & \mbox{and} & y=1, \mbox{and} & 2 \frac{\partial v}{\partial y} + \frac{\partial u}{\partial x} = 0 & \mbox{at} & y=1,
  \end{array}
\end{gather*}
where
\begin{equation*}
  \nu = \frac{1}{2}\left(\left(\frac{\partial u}{\partial x}\right)^2 + \left(\frac{\partial v}{\partial y}\right)^2 + \frac{1}{4} \left( \frac{\partial u}{\partial y} + \frac{\partial v}{\partial x} \right)^2 + \frac{\partial u}{\partial x} \frac{\partial v}{\partial y} \right)^{-\frac{n-1}{2n}}
\end{equation*}
is the velocity-dependent viscosity \cite{MacAyeal1989, MacAyeal1993, MacAyeal1997}.  Assuming that the surface elevation $s(x,y)$ and ice thickness $h(x,y)$ are known and that $n=\frac{1}{m} = 3$, the forward model maps realizations of the basal friction $\beta(x,y)$ to the horizontal velocities $u(x,y)$ and $v(x,y)$.

To define our Bayesian inference problem, we endow the log-basal friction field $\log \beta(x,y)$ with a Gaussian process prior, using an isotropic squared-exponential covariance kernel,
\be
C\left ( (x_1,y_1),(x_2,y_2) \right ) = \sigma^2 \exp \left( - \frac{(x_1-x_2)^2 + (y_1 - y_2)^2}{2l^2} \right),
\ee
with correlation length $l=0.1$ and variance $\sigma^2=25$. This field is easily parameterized with a Karhunen-Lo\`{e}ve (K-L) expansion \cite{Adler1981}:
\begin{equation*}
\beta (x, y; \theta) \approx \exp \left( \sum_{i=1}^d \theta_i \sqrt{\lambda_i} \varphi_i(x,y) \right),
\end{equation*}
where $\lambda_i$ and $\varphi_i(x,y)$ are the eigenvalues and eigenfunctions, respectively, of the integral operator on $[0,1]^2$ defined by the kernel $C$, and the parameters $\theta_i$ inherit independent standard normal priors, $\theta_i \sim \mathcal{N}(0,1)$. We truncate the Karhunen-Lo\`{e}ve expansion at $d=12$ modes and infer the weights $(\theta_1, \ldots, \theta_{12})$ from data. The true basal diffusivity field is shown in Figure \ref{fig:macayeal_friction}.

\begin{figure}
\centering
\includegraphics[scale=.85]{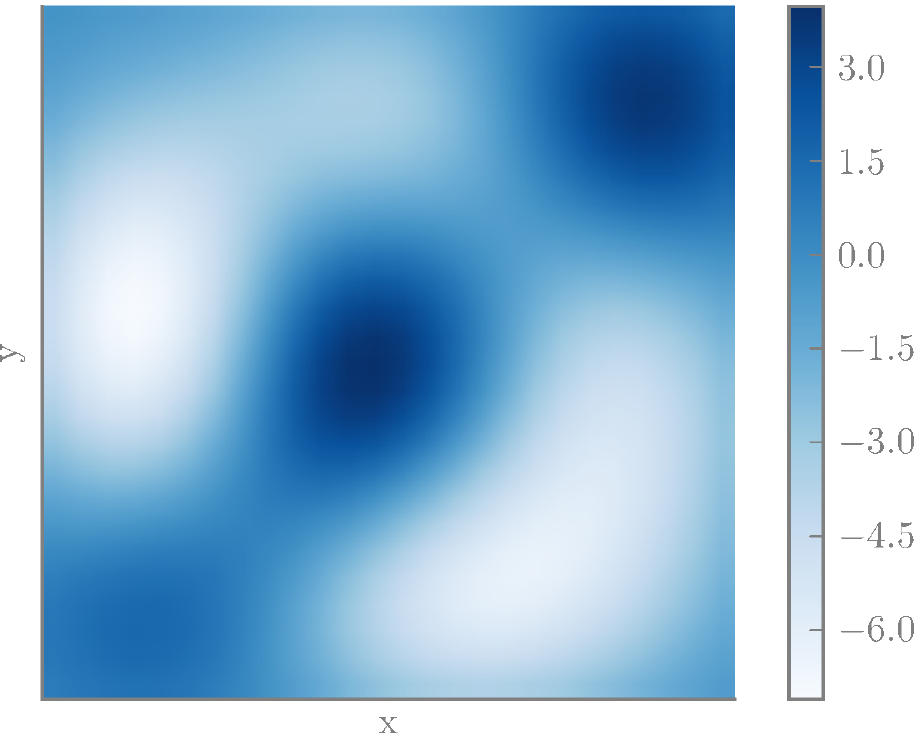}
\caption{Ice stream inference problem: the true log-basal friction field, $\log \beta(x,y)$.}
\label{fig:macayeal_friction}
\end{figure}

Data arise from observations of the velocity field on a uniform $10 \times 10$ grid covering the unit square, $(x_i, y_i) \in \{(.05, .05), \ldots, (.95, .95)\}$, as depicted in Figure \ref{fig:macayeal_obs}. Both the $u$ and $v$ components of velocity are observed, and observational errors are taken to be independent, additive, and identically Gaussian, $\mathcal{N}(0, 0.01^2)$. To avoid an \change{``inverse crime'' \cite{Kaipio2005}}, data are generated with a $25 \times 25$ mesh but inference uses a coarser $15 \times 15$ mesh.

\begin{figure}
\centering
\includegraphics[scale=.85]{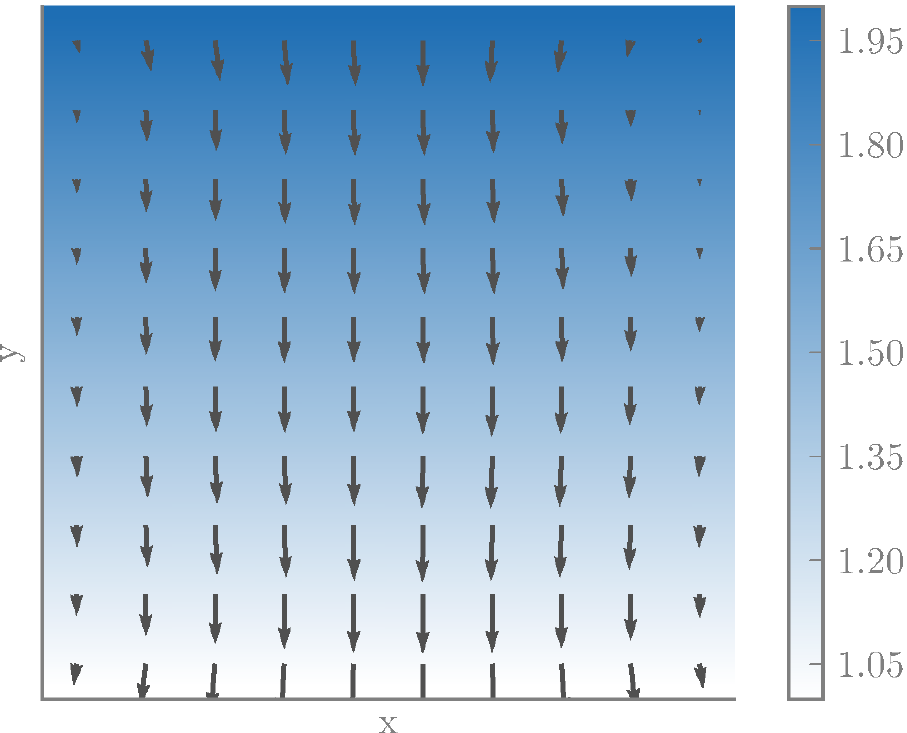}
\caption{Ice stream inference problem: the assumed ice height field $h(x,y)$ (shading) superimposed on the observed velocity field (vectors), given the basal friction in Figure \ref{fig:macayeal_friction}. Note the left--right asymmetry in the velocity field at the top of the domain, induced by the high friction region at the top right.}
\label{fig:macayeal_obs}
\end{figure}

The posterior distribution of the basal friction field  is quite challenging to sample, as the forward model requires, on average, 26 seconds per evaluation. Using a direct MCMC approach, a numerical simulation comprising 10 parallel chains of 200,000 steps each would therefore take nearly two months to run. \change{Using LA+AM on 10 parallel chains,} we complete exactly the same simulation in just over one day, a nearly 60-fold improvement in the run time. Representative one- and two-dimensional marginals of the posterior (focussing on only the first 6 of 12 dimensions) are shown in Figure \ref{fig:macayeal_density}. Note that several parameters are strongly correlated, and that many marginal distributions appear skewed and non-Gaussian. These two million samples were produced using only about 35,000 runs of the forward model. 

\begin{figure}
\centering
\includegraphics[width=.8\textwidth]{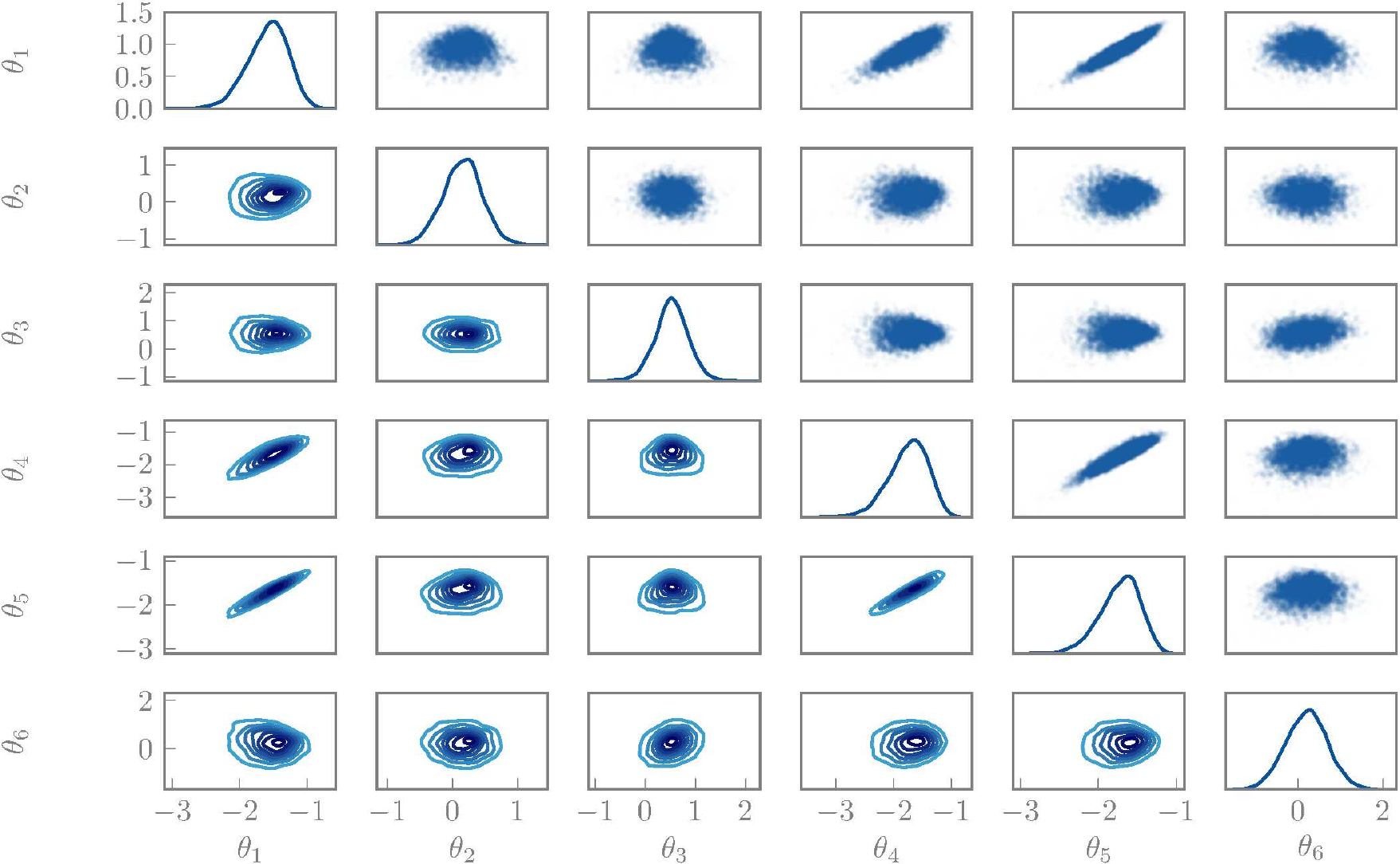}
\caption{One- and two-dimensional posterior marginals of the first six parameters in the ice stream inference problem.}
\label{fig:macayeal_density}
\end{figure}

\section{Conclusions}

This work has extended our previous development of asymptotically exact MCMC algorithms that employ local approximations of expensive models.
We lifted restrictive assumptions on the type of MCMC kernel that could be used---in particular, allowing the proposal distribution to extract derivatives, and hence geometric information, from the approximation. Doing so enables a wide variety of more sophisticated proposal distributions, such as manifold MALA, to be applied in settings where they would otherwise be intractable (e.g., when forward model derivatives cannot be directly evaluated) or unaffordable.
Additionally, we showed that using approximations allows the most computational intensive element of many MCMC simulations---the forward model or likelihood evaluations---to be directly parallelized, through the shared and online construction of a posterior-adapted set of samples. Sharing this set of model evaluations among multiple MCMC chains drives the construction of local approximations on each chain, providing a novel and effective means of reducing the run time of MCMC simulations. \change{Our shared local approximation scheme can readily be paired with other MCMC parallelization schemes, e.g., methods that use the presence of multiple chains to improve mixing; this a natural avenue for future work.}

To demonstrate the practical utility of these developments, we presented two challenging inference problems that we believe reflect scientifically interesting settings where forward models are necessarily expensive. Using parallel computing resources, we demonstrated a nearly two-order-of-magnitude improvement in the run time of a groundwater hydrology inference problem, and a roughly 60-fold reduction in the run time of an ice stream inference problem. These results suggest that our approach may help make a range of challenging Bayesian inference problems feasible. A reusable and open source implementation of this algorithm is available as part of the MIT Uncertainty Quantification (MUQ) library, \url{http://muq.mit.edu}.

\section*{Acknowledgments}

This work was supported in part by the Scientific Discovery through Advanced Computing (SciDAC) program of the US Department of Energy, Office of Science, Advanced Scientific Computing Research under award number DE-SC0007099 (P.\ Conrad, A.\ Davis, and Y.\ Marzouk), by the National Sciences and Engineering Research Council of Canada (A.\ Smith), and by the Office of Naval Research (N.\ Pillai).

\appendix

\section{Complete algorithm description}
\label{apx:code}
This appendix provides a complete description of the local approximation MCMC algorithm from \cite{Conrad2014}, here extended to MCMC proposals that also employ the approximation $\tilde{\mathbf{f}}$. We replicate necessary subroutines from \cite{Conrad2014}; for a full discussion and derivation of these methods, please see that paper. The sketch given in Algorithm \ref{alg:algSketchMala} of Section~\ref{sec:review} is here expanded into Algorithm \ref{alg:algOverviewGeneral}, which takes additional parameters $\beta_t$ and $\gamma_t$ that determine when refinement is performed according to random or cross validation criteria, respectively. The choice of $\gamma_t$ is arbitrary, but $\sum_t \beta_t$ must diverge; based on the parameter study in \cite{Conrad2014}, the numerical experiments in Section~\ref{sec:numerical} are performed with $\beta_t = 0.01t^{-0.2}$ and $\gamma_t = 0.1t^{-0.1}$. These numerical experiments employ local quadratic approximations, as described below. Code used to run the examples, in conjunction with MUQ, is provided in the Supplementary Material.


\algrenewcomment[1]{\hfill\makebox[0.33\linewidth][l]{\(\triangleright\) #1}}

\begin{algorithm}
\caption{Metropolis-Hastings with local approximations and general proposals}
\label{alg:algOverviewGeneral}
\begin{algorithmic}[1]

\Procedure{\textsc{RunChain}}{$\mathbf{f}, r, q, \theta_{1}, \mathcal{S}_{1}, \ell, \mathbf{d}, p, T, \{ \beta_t \}_{t=1}^{T}, \{ \gamma_t \}_{t=1}^{T}$}
\For{$t = 1 \ldots T$}
	\State  $(\theta_{t+1}, \mathcal{S}_{t+1}) \gets K_t(\theta_t, \mathcal{S}_t, \ell, \mathbf{d}, p, \mathbf{f}, r, q, \beta_t, \gamma_t)$
\EndFor
\EndProcedure

\Statex

\Procedure{$K_t$}{$\theta^-, \mathcal{S},  \ell, \mathbf{d}, p, \mathbf{f}, r, q, \beta_t, \gamma_t$}
	\State Draw proposal $\mathbf{z}_t \sim \mathcal{N}(0,I)$ 
	\State $\tilde{\mathbf{f}}^- \gets \textsc{LocApprox}(\theta^-, \mathcal{S}, \emptyset)$ \label{alg:algOverview:fplus} 

\State  $\theta^+ \gets r(\theta^-, \mathbf{z}_t, \tilde{\mathbf{f}}^-)$ 

	\State $\tilde{\mathbf{f}}^+ \gets \textsc{LocApprox}(\theta^+, \mathcal{S}, \emptyset)$  

	\State $\alpha \gets \min \left(1,
\frac{\ell(  \theta^+ | \mathbf{d} ,\tilde{\mathbf{f}}^+)p(\theta^+)q(\theta^+, \theta^- | \tilde{\mathbf{f}}^+)}{\ell(\theta^-  | \mathbf{d}  ,\tilde
{\mathbf{f}}^-)p(\theta^-)q(\theta^-, \theta^+ | \tilde{\mathbf{f}}^-)} \right)$ \Comment{Compute nominal acceptance ratio} \label{alg:algOverview:errorStart} 
\State Compute $\epsilon^+$ and $\epsilon^-$ as in \eqref{eq:cvStart}--\eqref{eq:cvEnd}.

\If{$u \sim \text{Uniform}(0,1) < \beta_m$} \Comment{Refine with probability $\beta_m$}
	\State Randomly, $\mathcal{S} \gets \textsc{RefineNear}(\theta^+, \mathcal{S})$ or $\mathcal{S} \gets \textsc{RefineNear}(\theta^-, \mathcal{S})$
\ElsIf{$\epsilon^+ \geq \epsilon^-$ and $ \epsilon^+ \geq \gamma_m$} \Comment{If needed, refine near the larger error}  \label{alg:algOverview:refineStart}
\State $\mathcal{S} \gets \textsc{RefineNear}(\theta^+, \mathcal{S})$
\ElsIf{$\epsilon^- > \epsilon^+$ and $ \epsilon^- \geq \gamma_m$}
\State $\mathcal{S} \gets \textsc{RefineNear}(\theta^-, \mathcal{S})$

	\EndIf
	\If{refinement occurred} repeat from Line \ref{alg:algOverview:fplus}.
\Else \Comment{Evolve chain using approximations}
\State Draw $u \sim \text{Uniform}(0,1)$. If $u < \alpha$, \textbf{return} $(\theta^+, \mathcal{S})$, else \textbf{return} $(\theta^-, \mathcal{S})$.
\EndIf
\EndProcedure

\end{algorithmic}
\end{algorithm}

Algorithm \ref{alg:supporting} provides several subroutines. The first, \textsc{LocApprox}, gathers the $N$ nearest neighbors from $\mathcal{S}_t$ to use in constructing the approximation; for quadratics, $N = \frac{\sqrt{d}(d+1)(d+2)}{2}$. The operator $\mathcal{A}^{\sim j}_{\mathcal{B}(\theta, R)}$ constructs the local approximation; in this work, it fits a quadratic (a degree-two polynomial) with least squares. The input $j$ facilitates cross validation and unless $j=\emptyset$, designates that the $j$th neighbor should be omitted. The second routine, \textsc{RefineNear}, solves a local optimization problem to choose a new point $\theta^\ast$ that is near $\theta$ but space-filling  overall; this point is used to enrich $\mathcal{S}_t$.

\algrenewcomment[1]{\hfill\makebox[0.4\linewidth][l]{\(\triangleright\) #1}}

\begin{algorithm}
\caption{Supporting algorithms}
\label{alg:supporting}
\begin{algorithmic}[1]

\Procedure{\textsc{LocApprox}}{$\theta, \mathcal{S}, j$}
\State Select $R$ so that $|\mathcal{B}(\theta,R)| = N$, where 
\Statex $\mathcal{B}(\theta,R) :=  \{(\theta_i,\mathbf{f}(\theta_i)) \in \mathcal{S} : \| \theta_i - \theta\|_2 \leq R\}$ \Comment{Select ball of points}
\State $\tilde{\mathbf{f}} \gets \mathcal{A}_{\mathcal{B}(\theta,R)}^{\sim j} $ \label{alg:algOverview:regress}\Comment{Fit local approximation}
\State \textbf{return} $\tilde{\mathbf{f}}$
\EndProcedure

\Statex
\Procedure{\textsc{RefineNear}}{$\theta, \mathcal{S}$}
\State Select $R$ so that $|\mathcal{B}(\theta,R)| = N$ \Comment{Select ball of points}
\State $\theta^\ast \gets \argmax_{\|\theta^\prime - \theta\| \leq R} \min_{\theta_i \in \mathcal{S}} \| \theta^\prime - \theta_i\| $\Comment{Optimize near $\theta$}
\State $\mathcal{S} \gets \mathcal{S} \cup \{\theta^\ast, \mathbf{f}(\theta^\ast) \}$ \Comment{Grow the sample set}
\State \textbf{return} $\mathcal{S}$
\EndProcedure

\end{algorithmic}
\end{algorithm}

Cross validation is used to estimate the error in the acceptance probability evaluated using the approximations. Define the nominal and leave-one-out variants of the approximations, for $j=1,\ldots, N$, as
\begin{eqnarray*}
\tilde{\mathbf{f}}^+ &= \textsc{LocApprox}(\theta^+, \mathcal{S}, \emptyset)  \qquad  \tilde{\mathbf{f}}^+_{\sim_j}  &=  \textsc{LocApprox}(\theta^+, \mathcal{S}, j)\\
\tilde{\mathbf{f}}^- &= \textsc{LocApprox}(\theta^-, \mathcal{S}, \emptyset)  \qquad  \tilde{\mathbf{f}}^-_{\sim_j}  &=  \textsc{LocApprox}(\theta^-, \mathcal{S}, j).
\end{eqnarray*}
Then compute the approximate posterior ratio and all the leave-one-out variants (here slightly modified from our original work to include the proposal densities), 
\begin{eqnarray*}
\zeta &:=& \frac{\ell(\theta^+|\mathbf{d}, \tilde{\mathbf{f}}^+)p(\theta^+)q(\theta^+, \theta^- | \tilde{\mathbf{f}}^+)}
{\ell(\theta^-|\mathbf{d}, \tilde{\mathbf{f}}^-)p(\theta^-)q(\theta^-, \theta^+ | \tilde{\mathbf{f}}^-)} \\
\zeta^{+,\sim j} &:=& \frac{\ell(\theta^+|\mathbf{d}, \tilde{\mathbf{f}}^+_{\sim_j})p(\theta^+)q(\theta^+, \theta^- | \tilde{\mathbf{f}}^+_{\sim_j})}
{\ell(\theta^-|\mathbf{d}, \tilde{\mathbf{f}}^-)p(\theta^-)q(\theta^-, \theta^+ | \tilde{\mathbf{f}}^-)}\\
\zeta^{-,\sim j} &:=& \frac{\ell(\theta^+|\mathbf{d}, \tilde{\mathbf{f}}^+)p(\theta^+)q(\theta^+, \theta^- | \tilde{\mathbf{f}}^+)}
{\ell(\theta^-|\mathbf{d}, \tilde{\mathbf{f}}^-_{\sim_j} )p(\theta^-)q(\theta^-, \theta^+ | \tilde{\mathbf{f}}^-_{\sim_j})}
\end{eqnarray*}
Finally, find the maximum difference between the $\alpha$ values computed using $\zeta$ and those computed using the leave-one-out variants $\zeta^{+,\sim j}$ and $\zeta^{-,\sim j}$, averaging over the forward and reverse directions. These are the error indicators:
\begin{align}
\epsilon^+ &:=& \underset{j}{\max} \left( \bigg|\min \left(1, \zeta \right) - \min \left(1,\zeta^{+,\sim j}\right)\bigg| + \left|\min \left(1, \frac{1}{\zeta}\right) - \min \left(1,\frac{1}{\zeta^{+,\sim j}}\right)\right| \right) \label{eq:cvStart}\\
\epsilon^- &:=& \underset{j}{\max} \left( \bigg|\min \left(1, \zeta \right) - \min \left(1,\zeta^{-,\sim j}\right)\bigg| + \left|\min \left(1, \frac{1}{\zeta}\right) - \min \left(1,\frac{1}{\zeta^{-,\sim j}}\right)\right| \right) \label{eq:cvEnd}.
\end{align}

\section{Proofs of the main results} \label{apx:proofs}
Throughout this section, we use the notation $f(x) = O(g(x))$ to mean that there exists some constant $0 < C < \infty$ so that $f(x) \leq C g(x)$. If the constant $C$ depends on an important parameter, we sometimes use that parameter as a subscript for emphasis; for example, $\frac{x^{2}}{p} = O_{p}(x^{4})$ for all fixed $p > 0$, but there is no constant $C< \infty$ so that $\frac{x^{2}}{p} \leq C x^{4}$ uniformly in $p > 0$.

For any pair of measures $\mu,\nu$ on a metric space $(\mathcal{X},d)$, denote by $\Pi(\mu,\nu)$ the collection of all pairs of random variables $(X,Y) \in \mathcal{X}^{2}$ that have marginal distributions $\mathcal{L}(X)= \mu$, $\mathcal{L}(Y) = \nu$. Recall that the \textit{Wasserstein metric} on measures on a metric space $(\mathcal{X}, d)$ is given by
\be 
W_{d}(\mu,\nu) = \inf_{(X,Y) \in \Pi(\mu,\nu)} \E[d(X,Y)].
\ee 
We also use the shorthand $W_{p} \equiv W_{\| \cdot \|_{p}}$ when $1 \leq p \leq \infty$. The \textit{total variation} distance between two probability measures $\mu,\nu$ is given by $\| \mu - \nu \|_{\mathrm{TV}} = W_{\rho}(\mu,\nu)$, where $\rho(x,y) \equiv \textbf{1}_{x  \neq y}$. The  \textit{mixing time} of a Markov chain $\{ Z_{t} \}_{t \geq 0}$ with stationary distribution $\pi$ on state space $\Omega$ is 
\be 
\tmix = \inf \{t \, : \, \sup_{Z_{0} = z \in \Omega} \| \mathcal{L}(Z_{t}) - \pi \|_{\mathrm{TV}} < \frac{1}{4} \}.
\ee

\paragraph{Proof of Theorem \ref{ThmMalaConv}}
Denote the diameter of $\Omega$ by $D_{\Omega}$ and the mixing time of $K_{\infty}$ by $\tau_{\mathrm{mix}}$; by parts 1 and 2 of Assumption \ref{DefSomeAssumptions}, respectively, $D_{\Omega}, \tau_{\mathrm{mix}} < \infty$. For $\epsilon > 0$,  let $\tau_{\epsilon} = \inf \{ t > 0 \, : \, \mathcal{S}_{t} \text{ is an } \epsilon-\text{cover of } \Theta \}$. By substituting $\tau_{\epsilon}$ for $\tau$ everywhere that it is used, the proof of Lemma B.4 of \cite{Conrad2014} shows that
\be \label{IneqCoveringTimeMala}
\P[ \tau_{\epsilon} < \infty] = 1
\ee 
for all $\epsilon > 0$.

Next, fix $S,T \in \mathbb{N}$ and $\psi, \delta, \varphi_{0} > 0$, and let $\epsilon = \epsilon(\delta)$ be the smaller of the values of $\epsilon(\delta)$ from inequalities \eqref{EqEpsDeltCoverDef}, \eqref{IneqOtherContAssump}. Let $\mathcal{F}_{T}$ be the $\sigma$-algebra $\sigma( \{ X_{t}, \mathcal{S}_{t} \}_{0 \leq t \leq T} )$. We will let $\{ Y_{t} \}_{t \geq T}$ be a Markov chain with transition kernel $K_{\infty}$ started at $Y_{T} = X_{T}$ and we will let $\{ Z_{t} \}_{t \geq T}$ be a Markov chain with transition kernel $K_{\infty}$ started at the distribution $\mathcal{L}(Z_{T}) = \pi$. We now describe a coupling of the three stochastic processes $\{ X_{t} \}_{T \leq t \leq T+S}$, $\{ Y_{t} \}_{T \leq t \leq T+S}$, and $\{ Z_{t} \}_{T \leq t \leq T+S}$. We couple $\{ Y_{t} \}_{T \leq t \leq T+S}$, $\{ Z_{t} \}_{T \leq t \leq T + S}$ so that 
\be \label{IneqMalaTvCoupDef}
\P[Y_{T+S} = Z_{T+S} | Y_{T}, Z_{T}] = \| \mathcal{L}(Y_{T+S} | Y_{T}) - \mathcal{L}(Z_{T+S} | Z_{T}) \|_{\mathrm{TV}}.
\ee 
At least one coupling with this property exists by the definition of the total variation distance; choose one such coupling arbitrarily. We then couple $\{ X_{t} \}_{T \leq t \leq T+S}$ to $\{ Y_{t} \}_{T \leq t \leq T+S}$ iteratively in $t$. Denote by $\tilde{X}$ the value that would be returned in the $(t)$th iteration of Algorithm \ref{alg:algOverviewGeneral} if Step 21 were ignored, and let $\mathbf{z}$ be the value obtained in Step 7. Then, $(X_{t+1}, Y_{t+1})$ can be coupled conditional on $(X_{t}, Y_{t}, \mathcal{S}_{t})$ so that 
\be \label{IneqMalaWassCoupDef}
\begin{aligned}
\E[ \| X_{t+1} - Y_{t+1} \|_{2}] \leq & \E[ \|X_{t+1} - \tilde{X} \|] + \E[ \|\tilde{X} - Y_{t+1} \|] \leq \delta + D_{\Omega} \P[\| \mathbf{z} \| > \varphi_{0}] + \frac{\psi}{S+1} \\ & +  \sup_{\theta, \theta' \in \Theta, \, \| \theta - \theta' \| < \| X_{t} - Y_{t} \|} W_{2} (K_{\mathcal{S}_{t}}(\theta, \cdot), K_{\infty}(\theta', \cdot)). 
\end{aligned}
\ee 
Such a coupling exists by inequality \eqref{IneqOtherContAssump} and the definition of the Wasserstein distance. By the `gluing' lemma (Chapter 1 of \cite{villani2008optimal}), it is possible to combine the couplings of $\{ X_{t}, Y_{t} \}_{T \leq t \leq T+S}$ and $\{ Y_{t}, Z_{t} \}_{T \leq t \leq T+S}$ into a single coupling $\{ X_{t}, Y_{t}, Z_{t} \}_{T \leq t \leq T+S}$ that satisfies both inequality \eqref{IneqMalaTvCoupDef} and also inequality \eqref{IneqMalaWassCoupDef} for all $T \leq t < T+S$. Under this coupling,
\be \label{IneqMiddleWassMala1}
W_{2}(Y_{T+S}, Z_{T+S}) &\leq D_{\Omega} \P[Z_{T+S} \neq Y_{T+S}] \\
&\leq D_{\Omega} 2^{-\lfloor \frac{S}{\tau_{\mathrm{mix}}} \rfloor}.
\ee 
Let $\eta_{0}$ be as in the requirements for \eqref{IneqContAssump}. By inequalities \eqref{EqEpsDeltCoverDef} and \eqref{IneqMalaWassCoupDef}, we have for $T \leq t < T+S$ that
\be 
\E[ \| X_{t+1} - &Y_{t+1} \|_{2} \, | \mathcal{F}_{T}] = \E[\| X_{t+1} - Y_{t+1} \|_{2}  \textbf{1}_{T \geq \tau_{\epsilon}} | \mathcal{F}_{T}] + \E[\| X_{t+1} - Y_{t+1} \|_{2}  \textbf{1}_{T < \tau_{\epsilon}} | \mathcal{F}_{T}] \\
&\leq \delta + D_{\Omega} \P[\| \mathbf{z} \| > \varphi_{0}] + \frac{\psi}{S+1} + \E[ \sup_{\theta, \theta' \in \Theta, \, \| \theta - \theta' \| < \| X_{t} - Y_{t} \|} W_{2} (K_{\infty}(\theta, \cdot), K_{\infty}(\theta', \cdot)) | \mathcal{F}_{T}] \\ 
&+ \E[ \sup_{\theta \in \Theta}  W_{2}(K_{\mathcal{S}_{t}}(\theta, \cdot), K_{\infty}(\theta,\cdot)) \textbf{1}_{T \geq \tau_{\epsilon}} | \mathcal{F}_{T}]+ D_{\Omega} \textbf{1}_{T < \tau_{\epsilon}} \\
&\leq \delta + D_{\Omega} \P[\| \mathbf{z} \| > \varphi_{0}] + \frac{\psi}{S+1}   + C \E[ \| X_{t} - Y_{t} \|_{2} \, | \mathcal{F}_{T}] +  D_{\Omega} \P[\| X_{t} - Y_{t} \| \geq \eta_{0} | \mathcal{F}_{T}] \\
&+ \delta  + D_{\Omega} \textbf{1}_{T < \tau_{\epsilon}} \\
&\leq D_{\Omega} \P[\| \mathbf{z} \| > \varphi_{0}] + \frac{\psi}{S+1} + 2\delta + (C  + \frac{D_{\Omega}}{\eta_{0}})  \E[ \| X_{t} - Y_{t} \|_{2} \, | \mathcal{F}_{T}] + D_{\Omega} \textbf{1}_{T < \tau_{\epsilon}}.
\ee 
Iterating this inequality over $T \leq t < T+S$ and recalling that $\|X_{T} - Y_{T} \|_{2} = 0$, 
\be \label{IneqMiddleWassMala2}
\E[ \| X_{T+S} - Y_{T+S} \|_{2} \, | \mathcal{F}_{T}] &\leq (2\delta + \frac{\psi}{S+1} + D_{\Omega} \P[\| \mathbf{z} \| > \varphi_{0}]) (C + \frac{D_{\Omega}}{\eta_{0}})^{S+1} \\
&+  D_{\Omega} \textbf{1}_{T < \tau_{\epsilon}}.
\ee

Combining inequalities \eqref{IneqMiddleWassMala1} and \eqref{IneqMiddleWassMala2}, 
\be 
W_{2}(X_{T+S},\pi) &\leq \E[ \| X_{T+S} - Z_{T+S} \|_{2} ]  \\
&\leq  D_{\Omega} 2^{-\lfloor \frac{S}{\tau_{\mathrm{mix}}} \rfloor} +  (2 \delta + \frac{\psi}{S+1} + D_{\Omega} \P[\| \mathbf{z} \| > \varphi_{0}]) (C + \frac{D_{\Omega}}{\eta_{0}})^{S+1} +  D_{\Omega} \P[ T < \tau_{\epsilon}].
\ee 
Letting $\psi$ go to 0, 
\be \label{IneqPenultimate}
W_{2}(X_{T+S},\pi)  & \leq D_{\Omega} 2^{-\lfloor \frac{S}{\tau_{\mathrm{mix}}} \rfloor} + (2\delta + D_{\Omega} \P[\| \mathbf{z} \| > \varphi_{0}]) (C + \frac{D_{\Omega}}{\eta_{0}})^{S+1} \\
&+  D_{\Omega} \P[ T < \tau_{\epsilon}].
\ee 
For $\alpha \in \mathbb{N}$, define $\delta(\alpha) = \frac{1}{\alpha^{2}}$, $\varphi_{0}(\alpha) = \inf \{ \varphi \, : \, \P[\| \mathbf{z} \| > \varphi] \leq \alpha^{-2} \}$,  $S(\alpha) = \lfloor \frac{-\log(\alpha)}{\log(C + \frac{D_{\Omega}}{\eta_{0}})} \rfloor - 1$, and $T(\alpha)' =  \inf \{ t \, : \, \P[t < \tau_{\epsilon(\delta(\alpha))}] \leq \frac{1}{\alpha} \}$. It is easy to check that $\lim_{\alpha \rightarrow \infty} S(\alpha) = \lim_{\alpha \rightarrow \infty} T(\alpha)' = \infty$, and so for any sequence $T(\alpha) > T(\alpha)'$ inequality \eqref{IneqPenultimate} implies
\be 
\lim_{\alpha \rightarrow \infty} W_{2}(X_{T(\alpha) + S(\alpha)}, \pi)  \leq \lim_{\alpha \rightarrow \infty} \big( D_{\Omega}2^{-\lfloor \frac{S(\alpha)}{\tau_{\mathrm{mix}}} \rfloor} + \frac{4 D_{\Omega}}{\alpha} \big) = 0.
\ee 
Since this holds for \textit{any} sequence $T(\alpha) > T(\alpha)'$,  inequality \eqref{IneqConcWassConvMala} follows.\footnote{Since the convergence to stationarity under the Wasserstein distance may not be monotone, this flexibility in the choice of $T(\alpha)$ is necessary to obtain the desired convergence result.} {\color{header1}\rule{1.5ex}{1.5ex}}


\medskip

\paragraph{Proof of Theorem \ref{thm:mala}}
It is enough to check that the conditions of Theorem \ref{ThmMalaConv} hold. Going through the elements of Definition~\ref{MalaAssumptions} in order:
\begin{enumerate}
\item To check that inequality \eqref{EqEpsDeltCoverDef} holds, fix $\delta > 0$. By results in \cite{Conn2009},\footnote{The required result is a combination of Theorems 3.14 and 3.16, as discussed in the text after the proof of Theorem 3.16 of \cite{Conn2009}.} there exists a constant $\epsilon_{1} = \epsilon_{1}(\delta, \lambda) > 0$ so that for all $\epsilon < \epsilon_{1}$
\be \label{MaleConfirmingAssumptionsContinuity1}
\sup_{\theta \in \Theta}|p_{\mathcal{S}}(\theta) - p(\theta | \mathbf{d}) | < \frac{\delta}{2 D_{\Omega}}
\ee
if $\mathcal{S}$ is  an $\epsilon$-cover and the points $\mathcal{B}(\theta, R)$ chosen in Step 2 of Algorithm \ref{alg:supporting} are $\lambda$-poised. The same discussion in \cite{Conn2009} implies that there exists a constant $\epsilon_{2} = \epsilon_{2}(\delta, \lambda) > 0$ so that for all $\epsilon < \epsilon_{2}$, 
\be \label{IneqToBeUsedForNextCondition}
\sup_{\theta \in \Theta} | M_{\mathcal{S}}(\theta) - M_{\infty}(\theta) |, \, |\nabla_\theta(\pi_{\mathcal{S}} ) - \nabla_\theta(\pi_{\infty} ) | < \delta
\ee 
if $\mathcal{S}$ is  an $\epsilon$-cover and the points $\mathcal{B}(\theta, R)$ chosen in Step 2 of Algorithm \ref{alg:supporting} are $\lambda$-poised.
Since the smallest singular value of $M(\theta)$ is bounded below uniformly in $\theta$, this implies that there exists a constant $\epsilon_{3} = \epsilon_{3}(\delta, \lambda) > 0$ so that for all $\epsilon < \epsilon_{3}$ (see \cite[Prop.~7]{Givens1984}),
%
%
\be \label{MaleConfirmingAssumptionsContinuity2}
\sup_{\theta \in \Theta} W_{2}(q_{\mathcal{S}}(\theta,\cdot), q_{\infty}(\theta,\cdot)) < \frac{\delta}{2 }
\ee  
as long as $\mathcal{S}$ is  an $\epsilon$-cover and the points $\mathcal{B}(\theta, R)$ chosen in Step 2 of Algorithm \ref{alg:supporting} are $\lambda$-poised.

Combining inequalities \eqref{MaleConfirmingAssumptionsContinuity1} and \eqref{MaleConfirmingAssumptionsContinuity2}, we have for all $0 < \epsilon < \min(\epsilon_{1},\epsilon_{3})$  that 
\be 
\sup_{\theta \in \Theta} W_{2} (K_{\mathcal{S}}(\theta,\cdot), K_{\infty}(\theta,\cdot)) &\leq \sup_{\theta \in \Theta} W_{2}(q_{\mathcal{S}}(\theta,\cdot), q_{\infty}(\theta,\cdot)) + D_{\Omega}  \sup_{\theta \in \Theta}|p_{\mathcal{S}}(\theta) - p(\theta | \mathbf{d}) |  \\
&\leq \frac{\delta}{2} + \frac{\delta}{2} = \delta. 
\ee 
This completes the proof of inequality \eqref{EqEpsDeltCoverDef}.
\item 
By the assumption that the mass matrix $M(\theta)$ and likelihood $\ell(\theta | \mathbf{d}, \mathbf{f})$ are both $C^{\infty}$ functions on $\Omega$, and that the smallest singular value of $M$ and the likelihood $\ell$ are both uniformly bounded away from zero, we have 
\be \label{ContCalc1}
\| q_{\infty}(\theta,\cdot) &- q_{\infty}(\theta',\cdot) \|_{\mathrm{TV}} = \left \Vert \, \mathcal{N}\left(\theta + \frac{\epsilon}{2} M(\theta) \nabla_\theta \log{ \left (\ell(\theta | \mathbf{d}, \mathbf{f}) p(\theta) \right )}, \epsilon  M(\theta)\right) \right . \\ & \qquad -  \left . \mathcal{N}\left(\theta' + \frac{\epsilon}{2} M(\theta') \nabla_\theta \log{ \left (\ell(\theta' | \mathbf{d}, \mathbf{f}) p(\theta') \right )}, \epsilon  M(\theta')\right)  \right \Vert_{\mathrm{TV}} \\
&\leq \left \Vert \, \mathcal{N}\left(\theta + \frac{\epsilon}{2} M(\theta) \nabla_\theta \log{ \left (\ell(\theta | \mathbf{d}, \mathbf{f}) p(\theta) \right )}, \epsilon  M(\theta)\right)  \right . \\ & \qquad -  \left . \mathcal{N}\left(\theta' + \frac{\epsilon}{2} M(\theta') \nabla_\theta \log{ \left (\ell(\theta' | \mathbf{d}, \mathbf{f}) p(\theta') \right )}, \epsilon  M(\theta)\right)  \right \Vert_{\mathrm{TV}} \\ 
&+ \left \Vert \, \mathcal{N}\left(\theta' + \frac{\epsilon}{2} M(\theta') \nabla_\theta \log{ \left (\ell(\theta' | \mathbf{d}, \mathbf{f}) p(\theta') \right )}, \epsilon  M(\theta)\right) \right .  \\ & \qquad -  \left . \mathcal{N}\left(\theta' + \frac{\epsilon}{2} M(\theta') \nabla_\theta \log{ \left (\ell(\theta' | \mathbf{d}, \mathbf{f}) p(\theta') \right )}, \epsilon  M(\theta')\right)  \right \Vert_{\mathrm{TV}} \\
&= O_{c}( \Vert \theta - \theta' \Vert ),
\ee 
where the bound on the first term in the last line is standard, and the second term in the last line is bounded by an application of \cite[Lem.~4.8]{Klartag2007}.
By a similar calculation,
\be \label{ContCalc2}
| \alpha_{\infty}(\theta,z) - \alpha_{\infty}(\theta',z) | = O_{c}(\| \theta - \theta' \|).
\ee 
Inequalities \eqref{ContCalc1} and \eqref{ContCalc2} imply that
\be 
\sup_{\theta, \theta' \in \Theta, \, \| \theta - \theta' \| < \eta} W_{2} & ( K_{\infty}(\theta,\cdot),  K_{\infty}(\theta',\cdot)) \leq D_{\Omega} \sup_{\theta, \theta' \in \Theta, \, \| \theta - \theta' \| < \eta}\| K_{\infty}(\theta,\cdot) -  K_{\infty}(\theta',\cdot) \|_{\mathrm{TV}} \\
&\leq D_{\Omega} (\sup_{\theta,\theta' \in \Theta, \, \| \theta - \theta' \| < \eta} \| q_{\infty}(\theta,\cdot) - q_{\infty}(\theta',\cdot) \|_{\mathrm{TV}} + \sup_{\theta, \theta',z \in \Theta, \, \| \theta - \theta' \| < \eta} | \alpha_{\infty}(\theta,z) - \alpha_{\infty}(\theta',z) | \big) \\
&= O( \| \theta - \theta' \|).
\ee
This completes the proof of inequality \eqref{IneqContAssump}.
\item Inequality \eqref{IneqOtherContAssump} follows immediately from \eqref{MaleConfirmingAssumptionsContinuity1} and \eqref{IneqToBeUsedForNextCondition}.
\item The first item in Assumption \ref{DefSomeAssumptions} holds by our assumption that $\Theta$ is the $d$-dimensional \change{hypercube}.
\item The second item in Assumption \ref{DefSomeAssumptions} has two parts. The first part, that $p(\cdot | \mathbf{d})$ has a $C^{\infty}$ density that is bounded away from zero uniformly in $\theta,$ is an assumption of our theorem. The second part, that $q(\theta, \cdot | \mathbf{f})$ has $C^{\infty}$ density that is bounded away from zero uniformly in $\theta, \mathbf{f}$, follows from the form of the mMALA proposal and the fact that the state space is compact. 
\item Items 3 through 6 in Assumption \ref{DefSomeAssumptions} are assumed in the statement of the theorem. 
\end{enumerate}
This completes the proof of the theorem.
{\color{header1}\rule{1.5ex}{1.5ex}}

\bibliographystyle{siam} 
\bibliography{WorksCited,library,library_manual} 

\end{document}